\documentclass{article}

\usepackage{microtype}
\usepackage{graphicx}
\usepackage{subfigure}
\usepackage{booktabs} % for professional tables

\usepackage{hyperref}
\usepackage{booktabs}

 \usepackage[accepted]{icml2022}

\usepackage{amsmath}
\usepackage{amssymb}
\usepackage{mathtools}
\usepackage{amsthm}

\usepackage[capitalize,noabbrev]{cleveref}

\theoremstyle{plain}
\newtheorem{theorem}{Theorem}[section]

\newtheorem{lemma}[theorem]{Lemma}

\theoremstyle{definition}

\theoremstyle{remark}

\newtheorem{observation}{Observation}

\usepackage[textsize=tiny]{todonotes}

\newcommand{\cec}{\textsc{ECC}}
\newcommand{\minecc}{\textsc{MinECC}}
\newcommand{\maxecc}{\textsc{MaxECC}}

\usepackage{amsfonts}
\usepackage{amsmath}
\usepackage{mathrsfs}  
\usepackage{graphicx}
\usepackage{booktabs}
\usepackage{enumitem}
\usepackage{bbm}
\usepackage{mathtools}
\DeclarePairedDelimiter\ceil{\lceil}{\rceil}

\providecommand{\mat}[1]{\boldsymbol{\mathrm{#1}}}%
\renewcommand{\vec}[1]{\boldsymbol{\mathrm{#1}}}

\DeclareMathOperator*{\maximize}{maximize}

\DeclareMathOperator{\argmin}{argmin}

\providecommand{\mA}{\ensuremath{\mat{A}}}

\providecommand{\mC}{\ensuremath{\mat{C}}}

\providecommand{\mY}{\ensuremath{\mat{Y}}}

\providecommand{\vb}{\ensuremath{\vec{b}}}
\providecommand{\vc}{\ensuremath{\vec{c}}}

\providecommand{\vx}{\ensuremath{\vec{x}}}
\providecommand{\vy}{\ensuremath{\vec{y}}}

\usepackage{xparse}

\usepackage{amsthm}
\newcommand{\pr}[1]{\mathbb{P}\left[ #1 \right]}
\newcommand{\prc}[2]{\mathbb{P}\left[ #1 \mid #2 \right]}
\newcommand{\eim}{e \in \mathcal{M}_Y}

\usepackage{subfigure}

\icmltitlerunning{Optimal LP Rounding and Linear-Time Approximation Algorithms for Clustering Edge-Colored Hypergraphs}

\begin{document}

\twocolumn[
\icmltitle{Optimal LP Rounding and Linear-Time Approximation Algorithms for Clustering Edge-Colored Hypergraphs}



\icmlsetsymbol{equal}{*}

\begin{icmlauthorlist}
	\icmlauthor{Nate Veldt}{yyy}
\end{icmlauthorlist}

\icmlaffiliation{yyy}{Department of Computer Science and Engineering, Texas A\&M University, College Station, Texas, USA}

\icmlcorrespondingauthor{Nate Veldt}{nveldt@tamu.edu}

\icmlkeywords{Hypergraphs, clustering, linear programming}

\vskip 0.3in
]



\printAffiliationsAndNotice{} % otherwise use the standard text.

\begin{abstract}
We study the approximability of an existing framework for clustering edge-colored hypergraphs, which is closely related to chromatic correlation clustering and is motivated by machine learning and data mining applications where the goal is to cluster a set of objects based on multiway interactions of different \emph{categories} or \emph{types}. We present improved approximation guarantees based on linear programming, and show they are tight by proving a matching integrality gap. Our results also include new approximation hardness results, a combinatorial 2-approximation whose runtime is linear in the hypergraph size, and several new connections to well-studied objectives such as vertex cover and hypergraph multiway cut.
\end{abstract}

\section{Introduction}
Partitioning a graph into well-connected clusters is a fundamental algorithmic problem in machine learning. A recent focus in the machine learning community has been to develop algorithms for clustering problems defined over data structures that generalize graphs and capture rich information and metadata beyond just pairwise relationships. One direction has been to develop algorithms for clustering and learning over hypergraphs~\cite{panli2017inhomogeneous,li2018submodular,fountoulakis2021local,Hein2013}, which can model multiway relationships between data objects rather than just pairwise relationships. Another recent focus has been to develop techniques for clustering edge-colored graphs (and hypergraphs), where edge colors represent the \emph{type} or \emph{category} of pairwise interaction modeled by the edge~\cite{amburg2020clustering,Bonchi2015ccc,Anava:2015:ITP:2736277.2741629,klodt2021color,xiu2022chromatic,amburg2022diverse}. Edge color labels are relevant for many different applications. In co-authorship (hyper)graphs, an edge color may represent the type of publication venue or discipline~\cite{amburg2020clustering,Bonchi2012ccc}, which can be used to better ascertain the academic field an author belongs to. Edges in brain networks may have color labels to indicate differences in co-activation patterns between brain regions~\cite{Crossley11583}. In temporal graph analysis, edges within the same time period can be associated with the same color~\cite{amburg2020clustering}, which provides information about how interaction patterns evolve over time. In addition to these settings, algorithms for edge-colored graphs and hypergraphs have been applied to cluster food ingredients co-appearing in different recipes (where edge colors indicate cuisine types)~\cite{amburg2020clustering,klodt2021color,xiu2022chromatic}, genes in biological networks (where edge colors indicate  gene interaction types)~\cite{Bonchi2012ccc}, and users of social media platforms like Facebook and Twitter (where edge colors indicate types of social circles)~\cite{klodt2021color,xiu2022chromatic}.

\textbf{The Edge-Colored Clustering problem.} 
\begin{figure}[t]
	\centering
	\includegraphics[width = .75\linewidth]{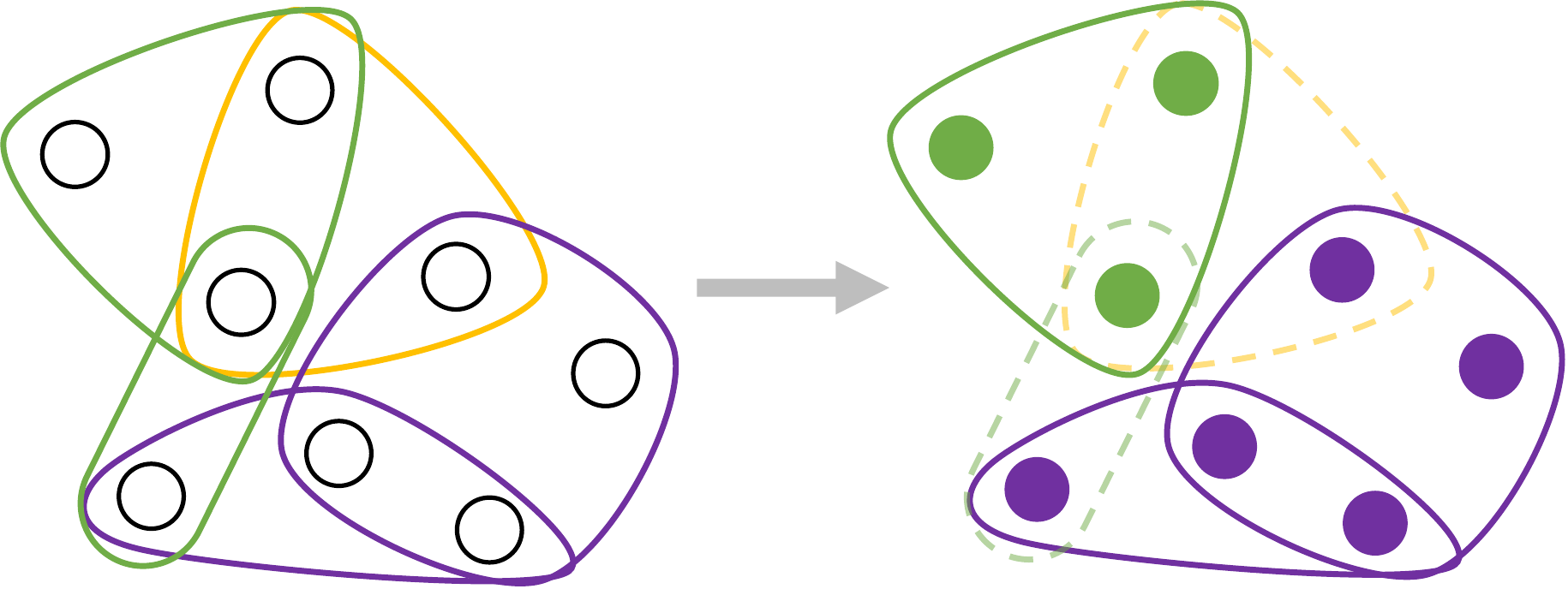}
	\vspace{-5pt}
	\caption{
		%		An instance of \textsc{Color-EC} is given by an edge-colored hypergraph. The goal is to color nodes in a way that minimizes the number (or weight) of \emph{unsatisfied} edges. In this example, 
		The best node coloring of this edge-colored hypergraph leaves three edges satisfied and two unsatisfied (dashed lines).}
	\label{fig:colorec}
	\vspace{-18pt}
\end{figure}
Edge-Colored Clustering (ECC) is a combinatorial optimization framework for clustering graphs and hypergraphs with edge colors.
This problem is also known by other related names such as \emph{colored clustering} or \emph{categorical edge clustering}. 
Let $H = (V,E)$ be a hypergraph where each hyperedge $e \in E$ is associated with one of $k$ colors for some $k \in \mathbb{N}$. The goal of ECC is to assign colors to \emph{nodes} in such a way that hyperedges tend to contain nodes that all have the same color as the hyperedge (see Figure~\ref{fig:colorec}). This can also be described as a clustering problem where one forms a single cluster of nodes for each color, and the goal is to do so in such a way that the cluster of nodes with color $c$ tends to contain hyperedges with color $c$. Note that the nodes in cluster $c$ do not necessarily need to be connected in the hypergraph in this formulation.
% Equivalently, the goal is to place nodes into clusters, where there is a single cluster for each color, in such a way that each cluster tends to include hyperedges of its color. 
To defined the objective more precisely, if all of the nodes in a hyperedge $e \in E$ are given the same color as $e$, then the hyperedge is \emph{satisfied}, otherwise it is \emph{unsatisfied} and we say the node color assignment has made a \emph{mistake} at this hyperedge. The goal is then to set node colors in a way that minimizes the number (or weight) of unsatisfied edges (the \minecc{} objective), or to maximize the number (or weight) of satisfied edges (the \maxecc{} objective). These are equivalent at optimality but different from the perspective of approximations. ECC is NP-hard~\cite{amburg2020clustering} but permits nontrivial approximation algorithms, whose approximation factors may depend on the number of colors $k$ and the rank $r$ of the hypergraph (the maximum hyperedge size).

\textbf{Background and previous results.} \citet{angel2016clustering} were the first to study the problem over graphs (the $r = 2$ case), focusing on \maxecc{}. They showed the objective is polynomial time solvable for $k = 2$ colors but NP-hard when $k \geq 3$, and provided a $\frac{1}{e^2}$-approximation algorithm by rounding a linear programming (LP) relaxation. A sequence of follow-up papers~\cite{ageev2014improved,ageev20200,alhamdan2019approximability} improved the best approximation factor to $4225/11664 \approx 0.3622$~\cite{ageev20200}, and showed it is NP-hard to approximate above a factor $241/249 \approx 0.972$~\cite{alhamdan2019approximability}. All of these results apply to \maxecc{} and $r = 2$. \citet{cai2018alternating} also focused on graphs but considered the problem from the perspective of parameterized complexity, showing that \minecc{} and \maxecc{} are fixed-parameter tractable (FPT) in the solution size.

\citet{amburg2020clustering} initiated the study of \cec{} in hypergraphs, focusing on \minecc{}. They gave a $\min \left\{2 - \frac1k, 2 - \frac1{r+1}\right\}$ -approximation algorithm based on rounding an LP relaxation, and provided combinatorial algorithms with approximation factors scaling linearly in $r$. Finally, they showed that the problem can be reduced in an approximation-preserving way to node-weighted multiway cut (\textsc{Node-MC})~\cite{garg2004multiway}. This leads to a $2(1-1/k)$ approximation by rounding the \textsc{Node-MC} LP relaxation. \citet{amburg2022diverse} later extended this framework for the task of diverse group discovery and team formation. Recently,~\citet{kellerhals2023parameterized} gave improved FPT algorithms and parameterized hardness results. % for graphs and hypergraphs.

\textbf{Other related work.} ECC is closely related to chromatic correlation clustering~\cite{Bonchi2015ccc,Anava:2015:ITP:2736277.2741629,klodt2021color,xiu2022chromatic}, an edge-colored generalization of correlation clustering~\cite{BansalBlumChawla2004}. The reduction to \textsc{Node-MC} also situates \minecc{} within a broad class of multiway partition problems~\cite{ene2013local,garg2004multiway,Dahlhaus94thecomplexity,calinescu2000,chekuri2011approximation,chekuri2011submodular,chekuri2016simple}, such as hypergraph multiway cut (\textsc{Hyper-MC}). Appendix~\ref{app:related} provides a deeper discussion of related work.

\textbf{Our contributions.}
There are many natural questions on the approximability of ECC that are unresolved in previous work, especially relating to the less studied \minecc{} objective.~\citet{amburg2020clustering} provide an indirect $2(1-1/k)$-approximation for \minecc{} by rounding the \textsc{Node-MC} LP relaxation in a reduced graph. When $k \leq 2(r+1)$, this is better than the approximation factor they obtain by rounding the canonical \minecc{} LP relaxation. An open question is to understand the exact relationship between these LP relaxations. In general, can we improve on the best LP rounding algorithms? For combinatorial algorithms, the best existing approximation factors scale linearly in $r$~\cite{amburg2020clustering}. Can we design a combinatorial algorithm whose approximation factor is constant with respect to $r$ and $k$? Finally, although the problem is known to be NP-hard and \maxecc{} is known to be APX-hard~\cite{alhamdan2019approximability}, can we say anything more about approximation hardness for \minecc{}? 

We answer all of these open questions, and in the process prove new connections to other well-studied combinatorial objectives. In terms of linear programming algorithms:
\vspace{-8pt}
\begin{itemize}[leftmargin=5pt, itemsep = -0em]
	\item We prove that the \minecc{} LP relaxation is always at least as tight as the \textsc{Node-MC} LP relaxation, and in some cases is strictly tighter (Theorem~\ref{thm:lps}).
	\item We improve the best approximation factors for rounding the \minecc{} LP relaxation from $\min \left\{2 - \frac1k, 2 - \frac1{r+1}\right\}$ to $\min \left\{2 - \frac2k, 2 - \frac{2}{r+1}\right\}$ (Theorems~\ref{thm:kthm}, \ref{thm:rthm}, and~\ref{thm:r2}).
	\item We prove a matching integrality gap, showing that our LP rounding schemes are in fact optimal (Lemma~\ref{lem:intgap}).
\end{itemize}
\vspace{-8pt}
 For the graph objective ($r = 2$), our approximation factor is $\frac43$, improving on the previous best $\frac53$-approximation~\cite{amburg2020clustering}.
% ~\citet{cai2018alternating} previously proved a 2-approximation for this graph case and posed as an open question the task of proving 
%  (after showing a 2-approximation). Our results provide a much stronger
%  and (especially considering the matching integrality gap result) provides an even stronger answer to the open question posed by~\citet{cai2018alternating} on the best approximation factor for this problem (these authors showed a 2-approximation).
% ~\citet{cai2018alternating} previously gave a 2-approximation for this problem, and posed the open question of determining whether improved approximations were possible. The $5/3$ approximation of~\citet{amburg2020clustering} already answered this affirmatively; our results improve this even further, providing a matching integrality gap lower bound. 
Our approximation result for $r = 2$ is the most challenging and in-depth result in our paper, and relies on a new technique for using auxiliary linear programs to bound the worst-case approximation factor obtained when rounding the \minecc{} LP. This proof requires a careful case analysis; we also formally prove why alternative strategies for trying to round the LP relaxation (which appear potentially easier at first glance) will in fact fail to provide a $\frac43$-approximation (Lemma~\ref{lem:badcases}). 
Aside from our linear programming results, our contributions include the following:
%In addition to our linear programming results:
\vspace{-8pt}
\begin{itemize}[leftmargin=5pt, itemsep = -0em]
	\item We prove that \textsc{Vertex Cover} is reducible to \minecc{} in an approximation preserving way (Theorem~\ref{thm:vc2cec}), implying APX-hardness. The details of our reduction allow us to further prove 
%	that it is UGC-hard to obtain an approximation factor that is below $2 - (2 + o_r(1))\frac{ \ln \ln r}{\ln r}$ for all possible values of $k$.
	UGC-hardness of approximation below $2 - (2 + o_r(1))\frac{ \ln \ln r}{\ln r}$ (for arbitrarily large $k$). 
	The reduction also implies hardness results for hypergraph \maxecc{}.
	\item We prove that hypergraph \minecc{} is reducible to \textsc{Vertex Cover} in an approximation preserving way (Theorem~\ref{thm:cec2vc}). Explicitly forming the reduced \textsc{Vertex Cover} instance lead to combinatorial 2-approximation algorithms that run in time $O(\sum_{v \in V} d_v^2 + |E|^2)$.
	\item We develop more careful 2-approximation algorithms that run in $O(\sum_{e \in E} |e|)$, i.e., linear in the hypergraph size.
\end{itemize}
%\vspace{-10pt}
Our results improve significantly on the previous best approximation guarantees for \minecc{} and are also tight or near-tight with respect to different types of lower bounds. Our LP rounding scheme is optimal in that in matches the integrality gap we show. All of our algorithms for hypergraph \minecc{} also have asymptotically optimal approximation factors assuming the unique games conjecture, since approximating \textsc{Vertex Cover} by a constant smaller than 2 is UGC-hard~\cite{khot2008vertex}. Our combinatorial algorithms also have asymptotically optimal runtimes, since it takes $\Omega(\sum_{e \in E} |e|)$ time simply to read the hypergraph input. In addition to our theoretical results, we confirm in numerical experiments that our algorithms are very practical.
% and make it possible to obtain high-quality approximations on a much larger scale than was previously possible. 
Our empirical contributions include a new large benchmark dataset for edge-colored hypergraph clustering and a very fast method that uses additional heuristics to improve the practical performance of our combinatorial algorithms.

 \section{\minecc{} and LP Framework} 
\label{sec:lpframework}
An instance of \minecc{} is given by an edge-colored hypergraph $H = (V,E, C, \ell)$ with node set $V$ and edge set $E$; $w_e \geq 0$ represents the nonnegative weight for $e \in E$. We use the term \emph{edge} even in the hypergraph setting (rather than \emph{hyperedge}), to use uniform terminology between the graph and hypergraph setting.  $C =  [k] = \{1,2, \hdots k\} $ is a set of edge colors and $r$ denotes the maximum hyperedge size. The function $\ell \colon  E \rightarrow C$ maps each edge to a color in $C$, and $E_c \subseteq E$ denotes the set of edges with color $c \in C$.

%The goal of \minecc{} is to assign each node a color in a way that minimizes the weight of hyperedge mistakes.
% where a mistake is any hyperedge containing a node whose color does not match the hyperedge color. 
Let $Y$ be a map from nodes to colors where $Y[v] \in C$ is the color of node $v$. For $e \in E$, if there is any node $v \in e$ such that $Y[v] \neq \ell(e)$, the map $Y$ has made a mistake at edge $e$, and this incurs a penalty of $w_e$. This is equivalent to partitioning nodes into $k$ clusters, with each cluster corresponding to one color, in a way that minimizes the weight of edges that are not completely contained in the cluster with a matching color. Given $Y$, let $\mathcal{M}_Y \subseteq E$ denote the set of edges where $Y$ makes a mistake. The objective is 
\begin{equation}
	\label{eq:cecobj}
	(\minecc) \;\;\;
{\textstyle	\min_{Y} \;  \sum_{e \in E} w_e \mathbbm{1}_{\mathcal{M}_Y}(e),}
\end{equation}
where $\mathbbm{1}_{\mathcal{M}_Y}$ is the indicator function for edge mistakes and the minimization is over all valid node colorings $Y$. The canonical LP relaxation for this objective is
\begin{equation}
	\label{eq:ceclp}
	\begin{array}{lll}
		\min & {\sum_{e\in E}} w_e x_e & \\ 
		\text{s.t.} & \sum_{i = 1}^k x_v^i = k -1 & \forall v \in V \\
		& x_e \geq x_v^c & \text{if $e \in E_c$ where $c \in C$} \\
		& 0 \leq x_v^i \leq 1 &\forall v \in V \text{ and }i \in C\\
		&  0 \leq x_e \leq 1& \forall e \in E.
	\end{array}
\end{equation}
The variable $x_u^i$ can be interpreted as the distance between node $u$ and color $i$. Every node color map $Y$ can be translated into a binary feasible solution for this LP by setting $x_v^i = 0$ if $Y[v] = i$ and $x_v^i = 1$ otherwise, and by setting $x_e = 1$ when $e \in \mathcal{M}_Y$ and $x_e = 0$ otherwise. The constraints $x_e \geq x_v^c$ for $e \in E_c$ and the nonnegativity of $w_e$ together imply that $x_e = \max_{v \in e} x_v^c$ at optimality.
% we will often use this relationships in the analysis of our approximation algorithms.
%In practice the constraint $x_e = \max_{v \in e} x_v^c$ can be replaced with  for each $v \in e$, which does not change the optimal solution. We leave this as $x_e = \max_{v \in e} x_v^c$ in our formulation to explicitly highlight a useful relationship between variables that holds at optimality and which we will frequently use when proving our approximation guarantees. 

\subsection{Generic LP Rounding Algorithm}
Algorithm~\ref{alg:gen} is our generic algorithm for rounding the \minecc{} LP relaxation, which takes an interval $I \subseteq [0,1]$ as input. It generates a random threshold $\rho \in I$, and identifies the set of nodes $v$ satisfying $x_v^i < \rho$ for each cluster $i$. If $x_v^i < \rho$, we say that color $i$ ``wants'' node $v$ and that $i$ is a candidate color for node $i$. A node may have more than one candidate color, so we generate a random permutation of $\{1,2, \hdots k\}$ that defines the priority of each color. In Algorithm~\ref{alg:gen}, a color has higher priority if it comes later in the permutation $\pi$. We then define $Y[v]$ to be the color that has the highest priority among all nodes that want $v$. Nodes that are not wanted by any color are given an arbitrary color.  
%we update $Y[v]$ for each color that wants $v$, so this value may be overwritten many times until it is assigned to the candidate color with highest priority. With this implementation, a color has higher priority if it comes later in the permutation. 

The rounding schemes of~\citet{amburg2020clustering} amount to running Algorithm~\ref{alg:gen} with a fixed threshold (i.e., $I$ is a single point). Here we use a random threshold from an interval $I$, similar to rounding strategies for convex relaxations of other multiway cut objectives~\cite{calinescu2000,ene2013local,chekuri2016simple}. Our goal is to bound the probability that $e \in \mathcal{M}_Y$ for an arbitrary $e \in E$. 
%If we can prove that for some constant $p$, we always have $\pr{\eim} \leq p x_e$, then the expected cost of Algorithm~\ref{alg:gen} is $\mathbb{E}\left[ \text{cost}(Y) \right] = \sum_{e \in E} w_e \pr{\eim} \leq p \sum_{e \in E} w_e x_e$
\begin{observation}
	\label{obs:approx}
	Let $Y$ be the output of Algorithm~\ref{alg:gen} for some interval $I \subseteq [0,1]$. If $\pr{e \in \mathcal{M}_Y} \leq p x_e$ for every $e \in E$, the output node coloring $Y$ is a $p$-approximate solution for \minecc{}, since the expected cost of $Y$ is
\begin{equation*}
	\mathbb{E}\Big[ \sum_{e \in E} w_e \mathbbm{1}_{\mathcal{M}_Y}(e) \Big] = \sum_{e \in E} w_e \cdot \pr{\eim} \leq p \sum_{e \in E} w_ex_e.
\end{equation*}
\end{observation}
\vspace{-10pt}
Given this observation, in order to prove approximation guarantees for Algorithm~\ref{alg:gen}, we will simply focus on bounding $\pr{e \in \mathcal{M}_Y}$ under different conditions on $I$, $r$, and $k$.
\begin{algorithm}[tb]
	\caption{\textsf{GenColorRound}($H,I$)}
	\label{alg:gen}
	\begin{algorithmic}[5]
		\STATE{\bfseries Input:} $H = (V,E,C, \ell)$, interval $I \subseteq [0,1]$
		\STATE {\bfseries Output:} Node coloring map $Y \colon V \rightarrow C$
		\STATE Solve the LP-relaxation~\eqref{eq:ceclp}
		\STATE $\rho \leftarrow$ uniform random value in $I$
		\STATE $\pi \leftarrow$ uniform random permutation of $\{1,2, \hdots, k\}$
		\STATE For $i \in \{1, 2, \hdots k\}$ define $S_i = \{v \in V : x_v^i < \rho \}$
		%		\FOR{$i = 1, 2, \hdots k$}
		%		\STATE Define $S_i = \{v \in V : x_v^c < \rho \}$
		%		\ENDFOR
		\FOR{$i = 1$ to $k$}
		\FOR{$v \in S_{\pi(i)}$}
		\STATE $Y[v] = \pi(i)$
		\ENDFOR
		\ENDFOR
		\STATE If $v \notin \bigcup_i S_i$, set $Y[v]$ to an arbitrary color
	\end{algorithmic}
\end{algorithm}
\begin{figure}[t]
	\centering
	\includegraphics[width = \linewidth]{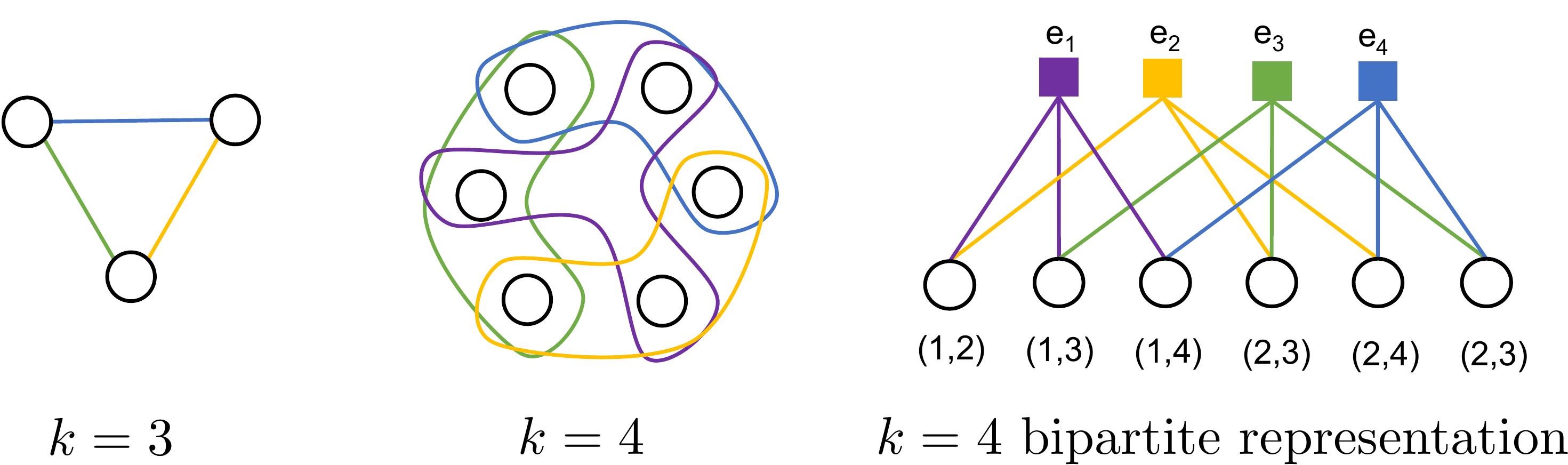}
	\vspace{-15pt}
	\caption{For integers $k \geq 3$ there is an edge-colored hypergraph for which the \minecc{} LP integrality gap is $2(1 - 1/k) = 2(1-1/(r+1))$. The instance involves $k$ edges of different colors, and for each pair of distinct edges there is one node in the intersection. 	
	This can be visualized using a bipartite representation: there is a hyperedge-node (squares) for each color, and each standard node (circles) is defined by pairs of hyperedge-nodes.
	}
	\label{fig:gap}
	\vspace{-15pt}
\end{figure}
\vspace{-10pt}

\subsection{Linear Program Integrality Gap}
Before proving new approximation guarantees, we establish a new integrality gap result for the canonical \minecc{} LP. This is related to integrality gap results for convex relaxations of \textsc{Node-MC}~\cite{garg2004multiway} and \textsc{Hyper-MC}~\cite{ene2013local}. For these problems, an integrality gap was given only in terms of the number of clusters $k$. For \minecc{}, we establish a gap in terms of $k$ and $r$ simultaneously by considering a hypergraph where $k = r+1$. See Figure~\ref{fig:gap} for an illustration and the appendix for a proof.
\begin{lemma}
	\label{lem:intgap}
	For every integer $k \geq 3$, there exists an instance $H = (V,E, C, \ell)$ of \minecc{} with $k = r+1$ whose optimal solution makes $k - 1 = r$ mistakes, and for which the LP relaxation has a value of $\frac{k}{2}= \frac{r+1}{2}$. Thus, the integrality gap is $2\big(1-\frac1k\big) = 2\big(1- \frac{1}{r+1}\big)$. 
\end{lemma}

\subsection{Comparison with the \textsc{Node-MC} LP Relaxation}
%Given a node-weighted graph $G = ({V},{E})$ with a distinguished set of $k$ terminal nodes, 
%the goal of \textsc{Node-MC} is to delete a minimum weight set of nodes to destroy all paths between terminals. 
\citet{amburg2020clustering} showed that \minecc{} reduces to \textsc{Node-MC} in an approximation-preserving way, implying a $2(1-1/k)$ approximation based on rounding the \textsc{Node-MC} LP relaxation. When $k \leq 2(r+1)$, this is better than the approximation guarantee the authors obtain for directly rounding the canonical \minecc{} LP. Despite this, we show that the canonical LP is tighter than the \textsc{Node-MC} LP. 
%This confirms that the discrepancy in previous approximation guarantees results because there is room for improvement in rounding schemes for the \minecc{} LP, and not because thinking about the problem in terms of \textsc{Node-MC} provides a stronger lower bound.
\begin{theorem}
	\label{thm:lps}
	The value of the \minecc{} LP relaxation is always at least as large as (and can be strictly larger than) the lower bound obtained by reducing to \textsc{Node-MC} and using the \textsc{Node-MC} LP relaxation.
\end{theorem}
Appendix~\ref{app:multiway} provides a more formal statement of this result and a proof, along with additional details on the relationship between \minecc{} and multiway cut objectives. In particular, we detail the approximation-preserving reductions from \minecc{} to \textsc{Node-MC} and \textsc{Hyper-MC}, and show that the best approximation results obtained via these reductions are in general not as strong as the approximation guarantees we develop using the \minecc{} LP relaxation.

\section{Optimal LP Rounding for Hypergraphs}
\label{sec:hypergraphs}
This section covers our optimal rounding scheme for the \minecc{} LP relaxation when $r > 2$, i.e., the hypergraph case. We provide our $\min \left\{2 - \frac2k, 2 - \frac{2}{r+1}\right\}$-approximation by combining two different approximation guarantees, obtained by running Algorithm 1 using two different intervals $I$. Most proofs are deferred to the appendix.

\textbf{Proof overview.}
Throughout the section we consider an arbitrary fixed edge $e$ with color $c = \ell(e)$ and LP variable $x_e = \max_{v \in e} x_v^c$. Recall that we say color $i$ ``wants'' node $v$ when $x_v^i < \rho$ for the randomly chosen threshold $\rho \in I$. 
In order to guarantee there is no mistake at $e$, color $c$ must want every node in $e$, which is true if and only if $\rho > x_e$. 
%This means that $\prc{e \in \mathcal{M}_Y}{\rho \leq x_e} = 1$.
 Even if $\rho > x_e$, there is a chance of making a mistake at $e$ if there are any other colors that want nodes from $e$. Our goal is to bound the probability of making a mistake at $e$ in terms of $x_e$. If $\pr{e \in \mathcal{M}_Y} \leq p x_e$ for some $p$, then Observation~\ref{obs:approx} guarantees Algorithm~\ref{alg:gen} is a $p$-approximation. 

\textbf{The color threshold value.} To aid in our results, we define the \emph{first color threshold} of $e$ to be 
$z_1^{(e)} = \min_{v \in e, i \neq c} x_v^i$.
%\begin{equation}
%	\label{eq:w1}
%	z_1^{(e)} = \min_{v \in e, i \neq c} x_v^i.
%\end{equation} 
This is the smallest value such that $\rho > z_1^{(e)}$ implies there exists a color other than $c$ that wants a node in $e$. 
\begin{observation}
	\label{xw}
	$1 - z_1^{(e)} \leq x_e$.
\end{observation}
To prove the observation, let $j \neq c$ and $v \in e$ be chosen so that $z_1^{(e)} = x_v^j$. The first constraint in the \minecc{} LP implies that $\sum_{i = 1}^k (1-x_v^i) = 1 \implies 2 - x_v^j - x_v^c \leq 1$, so $1 \leq x_v^j + x_v^c \leq z_1^{(e)} + x_e$. We use this observation in the proofs of Theorems~\ref{thm:kthm} and~\ref{thm:rthm}. 

To simplify notation, we can write $z_1$ instead of $z_1^{(e)}$, while still noting that this value is specific to edge $e$. We keep the 1 in the subscript of $z_1$ since this is the \emph{first} color threshold. We will define other color thresholds later for our rounding schemes for the graph case ($r = 2$), but this is not necessary for the results in this section.

%For Theorem~\ref{thm:kthm}, we additionally use that fact that if a color $i$ ``wants'' every node in an edge $e$, then the probability of making a mistake at $e$ is at most $(k-1)/k$, since there is a $1/k$ probability that color $c$ will have highest priority.
%\begin{equation}
%	\label{xw}
%	1 - z_1   \leq x_e.
%\end{equation}
\begin{theorem}
	\label{thm:kthm}
	Algorithm~\ref{alg:gen} with $I = (\frac12, \frac34)$ is a $2\left(1 - \frac1k\right)$-approximation for \minecc{}.
\end{theorem}
\vspace{-18pt}
\begin{proof}
	If $k = 2$, the constraint matrix for the LP relaxation is totally unimodular, so every basic feasible LP solution will have binary variables and our rounding procedure will find the optimal solution. Assume for the rest of the proof that $k \geq 3$. By Observation~\ref{obs:approx}, it suffices to show that $\pr{\eim} \leq 2 \left(1 - \frac{1}{k}\right) x_e$.
	%		\begin{equation}
	%		\label{kbnd}
	%		\pr{\eim} \leq 2 \left(1 - \frac{1}{k}\right) x_e.
	%	\end{equation}
	%
	If $x_e \geq \frac34$, then $\pr{\eim} \leq 1 \leq \frac43x_e \leq 2\left(1 - \frac1k\right)x_e$ for $k \geq 3$. We break up the remainder of the proof into two cases. 
	%	$x_e < \frac12$ and $x_e \in \left[\frac12, \frac34\right)$. 
	%	For remaining cases, we define a single color threshold value that indicates the point at which some color other than $c$ starts to want a node in $e$:
	%	Assume then that $x_e < 3/4$. Regardless 
	%	We know from Lemma~\ref{lem:zprob} that $\prc{\eim}{\rho > x_e} \leq ({k-1})/k$ regardless of the size of the hyperedge, since clearly there are never more than $k - 1$ colors other than $\ell(e)$ that can want a node from $e$. 
	%
	
	\textbf{Case 1:  $x_e < 1/2$.} 
	For $\rho \in (\frac12, \frac34)$, color $c = \ell(e)$ will want all nodes in $e$ because $x_e < \frac12$.
	Observation~\ref{xw}
	implies that $z_1 > \frac12$. If $z_1 \geq \frac34$, then no other color $i \neq c$ will want a node from $e$, meaning that $\pr{\eim} = 0$. If $\frac12 < z_1 \leq \frac34$, it is possible to make a mistake at $e$ only if $\rho > z_1$, and even then the probability of making a mistake is at most $\frac{k-1}{k}$, since there is a $\frac{1}{k}$ probability that color $c$ is given the highest priority by the random permutation $\pi$ in Algorithm~\ref{alg:gen}. Combining this with Observation~\ref{xw} shows
	\begin{align*}
		&\pr{\eim} =  \pr{\rho > z_1}\prc{\eim}{\rho > z_1} \\
		&\leq \frac{\frac34 - z_1}{{\frac34 - \frac12}} \cdot \frac{k-1}{k} =  4\left(\frac{k-1}{k}\right) \left( 1 - z_1 - \frac14\right) \\
		&\leq 4\left(1 - \frac1k\right)  \left(x_e - \frac{x_e}{2} \right) = 2\left(1 - \frac1k\right) x_e.
	\end{align*}
	%	We have used inequality~\eqref{xw}, assumption $x_e < \frac12$, and the fact that the probability of making a mistake at $e$ is always bounded above by $(k-1)/k$ and this can only happen if $\rho > z_1$.
	\textbf{Case 2:  $x_e \in \left[\frac12, \frac34\right)$.} 
	We apply a similar set of steps to see 	\begin{align*}
		\pr{\eim} &= \pr{\rho \leq x_e} \prc{\eim}{\rho \leq x_e}  \\ & \;\;\; + \pr{\rho > x_e} \prc{\eim}{\rho > x_e} \\
		&= \frac{x_e - \frac12}{{\frac34 - \frac12}}\cdot 1 + \frac{\frac34 - x_e}{{\frac34 - \frac12}} \cdot \frac{k-1}{k}  \\
		%		& = 4 \left( x_e - \frac{1}{2} + \frac{3(k-1)}{4k} - \frac{k-1}{k} x_e \right) \\
		&= 4 \left( \frac{x_e}{k} + \frac{k-3}{4k} \right) \leq 4 \left( \frac{1}{k} + \frac{k - 3}{2k} \right) x_e \\
		&= 2 \left(1 - \frac1k\right) x_e.
	\end{align*}
\end{proof}
\vspace{-10pt}
Theorem~\ref{thm:rthm} provides an approximation guarantee in terms of $r$, using a different interval. The main difference from Theorem~\ref{thm:kthm} is that we use the fact that $\rho < \frac23$ to bound the number of different colors that want a node in $e$, in terms of $r$. This allows us to bound the probability of making a mistake at $e$ in terms of $r$ instead of in terms of $k$.
\begin{theorem}
	\label{thm:rthm}
	Algorithm~\ref{alg:gen} with $I = (\frac12, \frac23)$ is a $2\big(1 - \frac{1}{r+1}\big)$-approximation for \minecc{} when $r > 2$.
\end{theorem}
Combining Theorems~\ref{thm:kthm} and~\ref{thm:rthm} provides an optimal rounding strategy, matching the integrality gap from Lemma~\ref{lem:intgap}.
\section{Optimal LP Rounding for Graphs}
\label{sec:graphs}
If $r = 2$, Theorem~\ref{thm:rthm} can be slightly adjusted to prove that when $I = (\frac12, \frac23)$, Algorithm~\ref{alg:gen} is a $\frac{3}{2}$-approximation for \minecc{}. 
Although this improves on the previous~$\frac53$-approximation~\cite{amburg2020clustering}, it does not match the integrality gap of $\frac43$ established by Lemma~\ref{lem:intgap}. This section shows how to use a larger interval $I = (\frac12, \frac78)$ to obtain a $\frac43$-approximation guarantee for the graph version ($r = 2$). The overall proof strategy mirrors our results for hypergraphs, but it requires several new technical results and a substantially more in-depth analysis. 
%Due to space constraints, all proofs are deferred to the appendix.
% In Theorem~\ref{thm:rthm}, we were able to use the fact that $\rho < \frac23$ to provide upper bounds on the number of colors wanting nodes in $e$, which were strong enough to prove the desired approximation factor without considering too many cases. When $\rho$ is only bounded above by $\frac78$, the bounds we get from a similar argument are much looser, and it becomes necessary to consider a much more in-depth case analysis. The $r = 2$ case therefore requires several new technical results and a much more in-depth analysis.
%We therefore begin with a detailed overview of our approach, followed by a systematic 
% Furthermore, we can show that using a smaller value than $\frac78$ than interval $I$ is bound to fail (see Section~\ref{sec:badcases} for more details on the challenges of the $r = 2$ case). 

\subsection{Proof Setup, Overview, and Challenges}
We consider a fixed arbitrary edge $e = (u,v)$ with color $c = \ell(e)$ and LP variable $x_e = \max\{x_u^c, x_v^c\}$.
We will show that for every {feasible} solution to the LP relaxation,
\begin{equation}
	\label{eq:main43}
	\pr{e \in \mathcal{M}_Y} \leq \frac43 x_e.
\end{equation}
By Observation~\ref{obs:approx}, this guarantees the $\frac43$-approximation. Inequality~\eqref{eq:main43} is easy to show for extreme values of $x_e$.
\begin{lemma}
	\label{lem:easyx}
	%	If $x_e \leq \frac18$ or $x_e \geq \frac34$, then if we apply Algorithm~\ref{alg:gen} with $I = (\frac12, \frac78)$, we have $\pr{e \in \mathcal{M}_Y} \leq \frac43 x_e$.
	If $x_e \notin (\frac18,\frac34)$, for Algorithm~\ref{alg:gen} with $I = (\frac12, \frac78)$, then $\pr{e \in \mathcal{M}_Y} \leq \frac43 x_e$.
\end{lemma}
\begin{proof}
	If $x_e \geq \frac34$, then $\pr{e \in \mathcal{M}_Y} \leq 1 \leq \frac43 x_e$. If $x_e \leq \frac18$, then for every $\rho \in (\frac12, \frac78)$, color $c = \ell(e)$ wants both $u$ and $v$, and the constraint $\sum_{j = 1}^k x_u^j = k-1$ can be used to show that no other color wants either $u$ or $v$, so $\pr{e \in \mathcal{M}_Y} = 0$.
\end{proof}
The following subsections cover the case $x_e \in (\frac18, \frac34)$, which is far more challenging. The difficulty in proving inequality~\eqref{eq:main43} for this case is that the probability of making a mistake depends heavily on the relationship among LP variables $x_e$ and $\{x_u^i, x_v^i \colon i \in [k]\}$, and more specifically the relative ordering of these variables. Since $k$ can be arbitrarily large, there can be many LP variables to consider, many possible orderings of these variables, and many different colors that want nodes $u$ and $v$. The best strategy for bounding $\pr{e \in \mathcal{M}_Y}$ depends on the configuration of LP variables, leading to an in-depth case analysis.

To provide a systematic proof for all cases, we first introduce a new set of variables $\{z_i\}_{i \in [k-1]}$ that can be viewed as a convenient rearrangement of the node-color distance variables $\{x_u^i, x_v^i \colon i \in [k]\}$. These generalize the \emph{first color threshold} used in Section~\ref{sec:hypergraphs}, and indicate the points at which there is a change in the number of distinct colors that ``want'' a node in $e$. We prove a key lemma on the relationship between $x_e$ and the $\{z_i\}_{i \in [k-1]}$ variables (Lemma~\ref{lem:xw}), and then show how to express $\pr{e \in \mathcal{M}_Y}$ in terms of these variables for different possible feasible LP solutions. We finally prove $\pr{e \in \mathcal{M}_Y} \leq \frac43x_e$ under different possible conditions by solving small auxiliary linear programs.
% whose optimal dual variables automatically determine a sequence of inequalities that bound $\pr{e \in \mathcal{M}_Y}$. 
%Throughout the proof, we highlight a number of reasons why it is challenging to develop a simpler proof that avoids this type of case analysis. 
%In Section~\ref{sec:badcases}, we also highlight why other intervals $I$ other than $\left(\frac12, \frac78\right)$ would necessitate and even more in depth analysis or would simply fail.

\subsection{The Color Threshold Lemma}
\label{sec:colorthresholds}
For $i \in \{0, 1, 2, \hdots, k-1\}$, the $i$th \textbf{color threshold} $z_i$ is the smallest nonnegative value such that for every $\rho > z_i$, there are at least $i$ distinct colors not equal to $c = \ell(e)$ that want a node in $e$.
This definition does not require the $i$ colors to all want the \emph{same} node in $e$. More precisely, for each of the $i$ colors, there exists a node $v \in e$ that is wanted by that color. By definition, the color threshold values are monotonic:
\begin{equation}
	\label{eq:wmonotone}
	0 = z_0 \leq z_1 \leq z_2 \leq \cdots \leq z_{k-1} \leq 1. 
\end{equation}
These values make it easier to express the probability of making a mistake at $e$ if we know $\rho > x_e$ and we know the value of $\rho$ relative to the color threshold values. Formally,
\begin{align}
	\label{prconditionedonw}
{\textstyle	\prc{e  \in \mathcal{M}_Y}{\rho > x_e \text{ and } z_i < \rho \leq z_{i+1}} = \frac{i}{i+1}.}
\end{align}
%One helpful way to interpret the color threshold values is as a rearrangement of the node-color variables $\{x_u^i, x_v^i \colon i \in [k]\}$. These generalize and extend the color threshold considered in Section~\ref{sec:hypergraphs}.
The following lemma is a generalization of Observation~\ref{obs:approx}, though its proof is more involved.
\begin{lemma}
	\label{lem:xw}
	For every integer $t \leq \frac{k}{2}$, %$t \leq x_e + z_t + z_{t+1} + \cdots + z_{2t-1}$.
		\begin{equation}
			\label{eq:keyinequality}
			t \leq x_e + z_t + z_{t+1} + \cdots + z_{2t-1}.
		\end{equation}
\end{lemma}

\subsection{Auxiliary LPs for Bounding Probabilities }
If we know where the color threshold values $\{z_i\}_{i \in [k-1]}$ are located relative to $x_e$ and the endpoints of $I = (\frac12, \frac78)$, we can derive an expression for $\pr{e \in \mathcal{M}_Y}$ in terms of these values when applying Algorithm~\ref{alg:gen}. The goal would then be to apply a sequence of carefully chosen inequalities to prove that this expression is bounded above by $\frac43x_e$. This subsection shows how to accomplish this task when we already know the relative ordering of these variables. In the next subsection, we apply this strategy to guarantee that this holds for \emph{all} possible orderings. 
%The appendix contains a proof of the following key lemmas, which
The following lemmas both rely on Lemma~\ref{lem:xw}.
\begin{lemma}
	\label{lem:lp12}
	If $x_e \in (\frac18, \frac12)$ and $z_{p-1} \leq \frac12\leq z_p \leq z_{q} \leq \frac78 \leq z_{q+1}$ for integers $p \leq q$, then $p = 1$, $q \leq 6$, and $\pr{e \in \mathcal{M}_Y} \leq \frac{8}{3}A_{q} x_e$ where $A_{q}$ is the optimal solution to
	\begin{equation}
		\label{lpA}
		%		\boxed{
		\begin{array}{lll}
			\textnormal{max} & {\displaystyle \frac{q}{q+1}\frac78\chi - \sum_{j = 1}^q \frac{1}{j(j+1)} \omega_j }\\
			\textnormal{s.t.} 
			&\textbf{(A\{i\}) } \omega_{i} - \omega_{i+1} \leq 0 \textnormal{ for $i = 1, \hdots, 5$ } \\
			&\textbf{(A6) } \chi - \omega_1 \leq 1 \\
			&\textbf{(A7) } 2\chi - \omega_2 - \omega_3 \leq 1 \\
			&\textbf{(A8) }3 \chi - 3\omega_5\leq 1 \\
			&\textbf{(A9) } -\chi \leq -2\;\;  \textbf{ (A10) } \omega_q - \frac78 \chi \leq 0. 
		\end{array}
		%		}
	\end{equation}	
\end{lemma}
We derive an analogous result for the case $x_e \in [\frac12, \frac34)$.
\begin{lemma}
	\label{lem:lp34}
	If $x_e \in [\frac12, \frac34)$ and $z_{p-1} \leq x_e \leq z_p \leq z_{q} \leq \frac78 \leq z_{q+1}$ for integers $p \leq q$, then $p \leq 5$, $q \leq 10$, and $\pr{e \in \mathcal{M}_Y} \leq \frac83 B_{p,q} x_e$ where $B_{p,q}$ is the solution to
		\begin{equation}
		\label{lpB}
		\begin{array}{lll}
			\textnormal{max}  & {\displaystyle \frac1p + \left(\frac{q}{q+1}\frac78 -\frac12\right) \chi - \sum_{j = p}^q \frac{1}{j(j+1)} \omega_j   }\\
			\textnormal{s.t.} 
			& \textbf{(B\{i\}) } \omega_{i} - \omega_{i+1} \leq 0  \textnormal{ for $i = 1, \hdots, 9$ } \\
			&\textbf{(B10) } \chi - \omega_1 \leq 1 \\
			&\textbf{(B11) } 2\chi - \omega_2 - \omega_3 \leq 1 \\
			&\textbf{(B12) }3 \chi - \omega_3 - \omega_4 - \omega_5\leq 1 \\
			& \textbf{(B13) }4 \chi - 4\omega_7 \leq 1 \;\; \textbf{ (B14) } \omega_{p-1} \leq 1 \\
			& \textbf{(B15) } -\omega_p \leq -1 \;\;\;\;\;  \textbf{ (B16) } \omega_q - \frac78 \chi\leq 0.
		\end{array}
	\end{equation}	
\end{lemma}

In order to prove that $\pr{e \in \mathcal{M}_Y} \leq \frac43 x_e$, it remains to show that for all valid choices of $p$ and $q$, the optimal solutions LP~\eqref{lpA} and LP~\eqref{lpB} are bounded above by $\frac12$. The result will then follow by Lemmas~\ref{lem:lp12} and~\ref{lem:lp34}. The difficulty is that there are many valid choices of $p$ and $q$ to check, each of which requires a different linear program. When $x_e \in (\frac18,\frac12)$, Lemma~\ref{lem:lp12} indicates that $p =1$ and $q \in \{1,2,3,4,5,6\}$, so we must consider six cases. When $x_e \in [\frac12,\frac34)$, Lemma~\ref{lem:lp34} shows that any pair $(p,q)$ satisfying $p \leq 5$, $q \leq 10$, and $p \leq q$ is a valid case, and there are 40 such pairs. If we are content with simply computing numerical solutions for each case, we can quickly confirm that the optimal solution for each of the 46 linear programs is bounded above by $\frac12$.
In the next subsection, we will show how to obtain a complete \emph{analytical} proof by using LP duality theory to extract a set of inequalities for each case that will prove an upper bound on the solution to each auxiliary LP. Note that there are many additional valid constraints that we could add to LPs~\eqref{lpA} and~\eqref{lpB}. We have omitted some constraints and have simplified others in order to obtain the smallest and simplest constraint set that suffices to prove $\pr{\eim} \leq \frac43x_e$. Adding more constraints makes the analysis more cumbersome without improving the bound.

%In Appendix~\ref{app:badcases} provides an in-depth discussion as to why it is challenging to avoid a lengthy case analysis.
%Before moving on, we remark that there are many other constraints that we could include in linear programs~\eqref{lpA} and~\eqref{lpB}, but choose to omit. 
%For example, the constraint $3 \chi - 4 \omega_5 \leq 1$ in linear program~\eqref{lpA} is a weaker version of the constraint that can be obtained by applying Lemma~\ref{lem:xw} with $t = 3$. For this linear program, we also do not enforce the constraint $4 \chi - \omega_4 - \omega_5 - \omega_6 \leq 1$, even though this is a valid constraint that results from applying Lemma~\ref{lem:xw} with $t = 4$. 
%We omit some constraints and simplify others in order to obtain the smallest and simplest constraint set that suffices to prove $\pr{\eim}/x_e \leq \frac43$. Adding more constraints makes the analysis more cumbersome, without making the worst case upper bound any smaller.

%We end this subsection with a brief discussion of the challenges in avoiding this lengthy case analysis. 
%
%that only requires us to compute and store one row of numbers for each case. Each row of numbers corresponds to dual variables for the case's linear program, which can be used to quickly extract a linear combination of inequalities (corresponding to LP constraints) which, when simply added together, will prove the bound on the probability.

\subsection{Using LP Duality for Case Analysis}
\label{sec:caseanalysis}
For each valid choice of $q$, we can use LP duality to bound the optimal solution of LP~\eqref{lpA} above by $\frac12$. The same basic steps also work for LP~\eqref{lpB}. The dual of LP~\eqref{lpA} is another LP with a variable for each constraint $\textbf{A\{i\}}$. By LP duality theory, every feasible solution to the dual provides an upper bound on the primal LP objective, so it suffices to find a set of dual variables with a dual LP value of $\frac12$ or less.

%In order to show that the primal LP solution is at most $\frac12$ in all cases, we simply need to produce an appropriate set of dual variables for each choice of $q$. 
%For different values of $q$, the objective for LP~\eqref{lpA} changes, 
%but we have almost the same set of constraints in each case. The only difference is that constraint \textbf{(A11)} depends on $q$, though we can effectively treat this as the same constraint in presenting dual varuable. 
%but we just 
%This means providing a set of 10 dual variables for each $q$, one variable for each constraint $\textbf{(A\{i\})}$ for $i = 1, 2, \hdots , 10$. The desired bound on the LP solution, and hence a bound on $\frac{8\pr{M}}{3x_e}$, is obtained simply by multiplying primal LP constraints by dual variables. 

\textbf{Proof for $q = 1$.} We illustrate this for LP~\eqref{lpA} when $q = 1$. The primal LP objective for this case is $\frac{7}{16} \chi - \frac12 \omega_1$, and we must bound this above by $\frac12$. Computing an optimal solution for the dual LP, we find that the constraint $\textbf{(A6)}$ corresponds to a dual variable of $\frac12$,  constraint $\textbf{(A9)}$ has a dual variable of $\frac{1}{16}$, and all other dual variables are zero. Multiplying constraints by the dual variables and summing produces a sequence of inequalities that proves the desired result:
\begin{align*}
	& \text{$\textbf{(A6)}$ $\chi - \omega_1 \leq 1$ and $\textbf{(A9)}$ $-\chi \leq -2$}:\\
	\implies &\frac12  (\chi  - \omega_1) + \frac{1}{16}(-\chi ) \leq    \frac12 ( 1 ) + \frac{1}{16} (-2)\\
 \implies &\frac{7}{16} \chi - \frac12 \omega_1 \leq \frac38 \leq \frac12 \text{ for all valid $\chi, \omega_1$.}
\end{align*}
Table~\ref{tab:lp12} in Appendix~\ref{app:dual} provides dual variables for LP~\eqref{lpA} for all other choices of $q$, and outlines a similar strategy for bounding LP~\eqref{lpB}. We conclude the following result.
% The appendix also shows how to use the same strategy to confirm that the optimal solution of LP~\eqref{lpB} is at most $\frac12$ for all valid choices of $p$ and $q$ given in Lemma~\ref{lem:lp34}. 
\begin{lemma}
	\label{lem:hardx}
	For any integer $q \in \{1,2, \hdots, 6\}$, the optimal solution to linear program~\eqref{lpA} is at most $\frac12$, and for any pair of integers $(p,q)$ satisfying $1 \leq p \leq 5$ and $p \leq q \leq 10$, the optimal solution to linear program~\eqref{lpB} is at most $\frac12$. Hence, for every $x_e \in (\frac18, \frac34)$, $\pr{\eim} \leq \frac43 x_e$ where $Y$ is the node coloring returned by Algorithm~\ref{alg:gen} when $I = (\frac12, \frac78)$.
\end{lemma}

%See the appendix for more complete details for this result, including optimal dual variables for each of 46 auxiliary LPs required to check all valid choices of $p$ and $q$. We also provide details for proving bounds more succinctly using a few matrix-vector products for each case. 
%In the appendix, we provide further details for how to prove the bounds in the last column of Table~\ref{tab:lp12} hold by simply checking a few matrix-vector products. We also provide a list of dual variables for the LP in Figure~\ref{fig:lpB}, for all 40 valid choices of $(p,q)$, which can similarly be used to prove the desired upper bound.
% In the appendix, we provide a similar table for the LP in Figure~\ref{fig:lpB} that lists dual variables for all 40 valid choices of $(p,q)$. 
%Combining all of these pieces yields our desired approximation guarantee.
Our main approximation result for $r = 2$ is now just a corollary of Lemmas~\ref{lem:easyx} and~\ref{lem:hardx} and Observation~\ref{obs:approx}.
\begin{theorem}
	\label{thm:r2}
	Algorithm~\ref{alg:gen} with $I = (\frac12, \frac78)$ is a $\frac43$-approximation algorithm for \textsc{Color-EC}.
\end{theorem}
The appendix details several  challenges that would need to be overcome in order to avoid a lengthy case analysis when trying to prove a $\frac43$-approximation.
%we can provide a more direct proof of Theorem~\ref{thm:r2} by ruling out or simultaneously handling a large number of cases for Lemmas~\ref{lem:lp12} and~\ref{lem:lp34} (rather than considering 46 auxiliary LPs). In the appendix, we detail several reasons why this appears challenging. 
Among other challenges, we prove the following result, indicating that if we apply Algorithm~\ref{alg:gen} with a different interval $I$, it would either become more challenging or impossible (depending on the interval $I$) to prove inequality~\eqref{eq:main43}.
% by using Algorithm~\ref{alg:gen} with a different interval $I$.
\begin{lemma}
	\label{lem:badcases}
	Let $Y$ denote the node color map returned by running Algorithm~\ref{alg:gen} with $I = (a,b) \subseteq (0,1)$. If $a > \frac12$, there exists a feasible LP solution such that $\pr{\eim} > \frac{4}{3} x_e$. If $a \leq \frac12$ and $b < \frac78$, there exists a feasible LP solution such that $\pr{\eim} > \frac43 x_e$. 
\end{lemma}

\section{Vertex Cover Equivalence and Algorithms}
\label{sec:vc}
\citet{cai2018alternating} showed that when $r = 2$, \minecc{} can be reduced in an approximation-preserving way to \textsc{Vertex Cover}. We not only extend this result to $r > 2$, but also prove an approximation-preserving reduction in the other direction that only applies if we consider the hypergraph ($r > 2$) setting. This immediately leads to new combinatorial approximation algorithms and refined hardness results. Explicitly forming the reduced \textsc{Vertex Cover} instance and applying existing \textsc{Vertex Cover} algorithms leads to runtimes that scale superlinearly in terms of $\sum_{e \in E} |e|$ (the size of the hypergraph). As a key algorithmic contribution, we design careful implicit implementations of existing \textsc{Vertex Cover} algorithms to provide 2-approximations for \minecc{} that have a runtime of $O(\sum_{e \in E}|e|)$.

\subsection{Vertex Cover Equivalence}
The equivalence between \emph{hypergraph} \minecc{} and \textsc{Vertex Cover} is summarized in two theorems (see Figure~\ref{fig:vcreductions}). 
\begin{figure}[t]
	\centering
	\includegraphics[width = \linewidth]{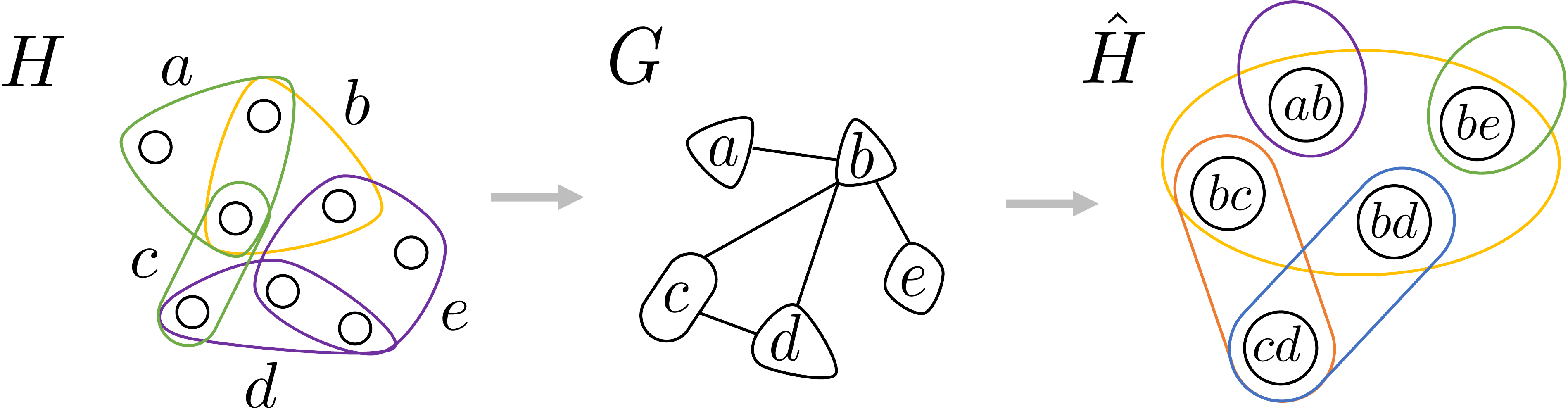}
	\vspace{-15pt}
	\caption{\minecc{} can be reduced to \textsc{Vertex Cover} by replacing each hyperedge with a node and adding an edge between nodes that represent overlapping hyperedges of different colors. \textsc{Vertex Cover} on a graph $G$ can be reduced to \minecc{} by introducing a node for each edge in $G$, and a hyperedge of a unique color for each node-neighborhood in $G$. 
%		When reducing from \minecc{} to \textsc{Vertex Cover} and then back ($H \rightarrow G \rightarrow \hat{H}$), the starting hypergraph $H$ and ending hypergraph $\hat{H}$ are typically not the same.
	}
	\label{fig:vcreductions}
	\vspace{-15pt}
\end{figure}
	\vspace{-10pt}
\begin{theorem}
	\label{thm:vc2cec}
	Let $G = (V,E)$ be a graph with maximum degree $\Delta$. \textsc{Vertex Cover} on $G$ can be reduced in an approximation-preserving way to \minecc{} on an edge-colored hypergraph $H$ with rank $r = \Delta$.
%	There exists an edge-colored hypergraph $H$ with maximum degree $\Delta$ such that \textsc{Vertex Cover} on $G$ is approximation-equivalent to \minecc{} objective on $H$.
\end{theorem}
\begin{theorem}
	\label{thm:cec2vc} 
	An instance $H = (V,E, C,\ell)$ of \minecc{} can be reduced in an approximation-preserving way to an instance of \textsc{Vertex Cover} on a graph $G$ with $|E|$ nodes.
%	Let $H = (V,E, C,\ell)$ be an edge-colored hypergraph. There exists a graph $G$ with $|E|$ nodes such that the \minecc{} objective on $H$ is approximation-equivalent to \textsc{Vertex Cover} on $G$.
\end{theorem}
Theorem~\ref{thm:vc2cec} implies new hardness results: assuming $r$ and $k$ are arbitrarily large, it is NP-hard to approximate hypergraph \minecc{} to within a factor better than 1.3606~\cite{dinur2005hardness}, and UGC-hard to approximate to within a factor that is a constant amount smaller than 2~\cite{khot2008vertex}. The relationship between maximum degree in $G$ and hypergraph rank in $H$ implies furthermore that it is UGC-hard to obtain a $2 - (2 + o_r(1))\frac{ \ln \ln r}{\ln r}$ approximation for rank-$r$ \minecc{} (assuming arbitrarily large $k$), using the \textsc{Vertex Cover} hardness results of~\citet{austrin2011inapproximability}. This result also implies that the maximum independent set problem is reducible in an approximation-preserving way to hypergraph \maxecc{}, indicating that it is NP-hard to obtain an approximation for hypergraph \maxecc{} that is sublinear in terms of the number of edges in the hypergraph~\cite{zuckerman2006linear}.

Theorem~\ref{thm:cec2vc} generalizes the reduction from \emph{graph} \minecc{} to \textsc{Vertex Cover} shown by~\citet{cai2018alternating}. This can be seen by viewing \minecc{} as an edge deletion problem. In an edge-colored hypergraph $H = (V,E,C, \ell)$, we say that $(e,f) \in E \times E$ is a \emph{bad edge pair} if $e$ and $f$ overlap and have different colors. The \minecc{} objective is then equivalent to deleting a minimum weight set of edges so that no bad edge pairs remain. A simple 2-approximation for \minecc{} can be designed by explicitly converting the edge-colored hypergraph into a graph and applying existing combinatorial \textsc{Vertex Cover} algorithms. Visiting each node $v \in V$ and then iterating through all pairs of hyperedges incident to $v$ in $H$ takes $O(\sum_{v \in V} d_v^2)$ time where $d_v$ is the degree of node $v$. The fastest algorithms for (unweighted \emph{and} weighted) \textsc{Vertex Cover} take linear-time in terms of the number of edges in the graph~\cite{pitt1985simple,bar1985local}. There can be up to $O(|E|^2)$ edges in the reduced graph $G$, so this approach has a runtime of $O(|E|^2 + \sum_{v \in V} d_v^2)$. 

\subsection{Linear-time 2-Approximation Algorithm}
\label{sec:linearvc}
It would be ideal to develop an algorithm for~\minecc{} whose runtime is not just linear in terms of the number of edges in a reduced graph, but is linear \emph{in terms of the hypergraph size}, i.e., $O(\sum_{e \in E} |e|)$. This may seem inherently challenging, given that existing \textsc{Vertex Cover} algorithms rely on explicitly visiting all edges in a graph. However, this can be accomplished using a careful \emph{implicit} implementation of a \textsc{Vertex Cover} algorithm. Pitt's algorithm~\cite{pitt1985simple} is a simple randomized 2-approximation for weighted \textsc{Vertex Cover} that iteratively visits edges in a graph $G = (V_G,E_G)$ in an arbitrary order. Each time it encounters an uncovered edge, it samples one of the two endpoints (in proportion to the node weight) to include in the cover. We will design a variant of Pitt's algorithm that we can apply directly to a \minecc{} instance $H$ without ever forming $G$.

\textbf{Additional notation.} 
Assume a fixed ordering of edges $E = \{e_1, e_2, \hdots , e_{|E|}\}$. For notational simplicity, let $w(i) = w_{e_i}$ and $\ell(i) = \ell(e_i)$ denote the weight and color of the $i$th edge, respectively. For each node $v \in V$, using a slight abuse of notation let $v(j)$ denote the index of the $j$th edge that is adjacent to $v$. These edge indices will be stored in a list $L_E(v) = [v(1), \;\; v(2), \;\; \cdots ,\;\; v(d_v)]$ where $d_v$ is the degree of $v$. For example, if node $v$ is in three hyperedges $\{e_2, e_5, e_9\}$, then $L_E(v) = [2, \; 5, \; 9]$. 
%Similarly, for every edge $e \in E$ we assume we can access a list of nodes in $e$. 
%Accessing the color $\ell(e)$ of an edge takes $O(1)$ time. 
We assume that edges are ordered by color, so that in $L_E(v)$, all of the indices for edges of color $1$ come first, then edges with color $2$, etc (see Figure~\ref{fig:pittstep}). If the hypergraph is not stored this way, it can easily be re-arranged in $O(\sum_{e \in E} |e|)$ time to satisfy this property. Our main goal is to obtain a set $\mathcal{D}$ of edge indices to delete that implicitly corresponds to an approximate vertex cover in $G$. Once we have obtained such an edge set $\mathcal{D}$, we can iterate through all edges in $H$, and for each $e \in E - \mathcal{D}$ we can assign all nodes in $e$ to have color $\ell(e)$, which takes $O(\sum_{e \in E} |e|)$ time. It remains to show how we can efficiently obtain a set $\mathcal{D}$ that implicitly encodes a 2-approximate vertex cover in $G$. 

\textbf{Implicit implementation.} 
Iterating through edges in $G$ is equivalent to iterating through bad edge pairs in $H$. A bad edge pair is ``covered'' if one of the edges is added to $\mathcal{D}$. Pitt's algorithm visits edges in $G$ in an arbitrary order, so we can traverse bad edge pairs in $H$ in any way that is convenient. We choose to iterate through nodes, and for each node $v \in V$ we will consider all bad edge pairs that contain $v$ in their intersection. Explicitly visiting all pairs of edges incident to $v$ takes $O(d_v^2)$ time. However, deleting an edge will cover multiple bad edge pairs at once, so we will be able to ``skip'' many pairs to speed up the process.
%This can be accomplished in $O(d_v)$ time by starting from both sides of the list $L_E(v) = [v(1), \;\; v(2), \;\; \cdots ,\;\; v(d_v)]$ and then iteratively moving inward while visiting bad hyperedge pairs. 
To accomplish this, we maintain pointers to the front ($f)$ and back ($b$) of $L_E(v)= [v(1), \;\; v(2), \;\; \cdots ,\;\; v(d_v)]$ which we initialize to $f = 1$ and $b = d_v$. Unless $v$ is only adjacent to nodes of one color, in the first step we know $(e_{v(f)}, e_{v(b)})$ will be a bad edge pair that needs to be covered. Using Pitt's technique, we randomly sample one edge to delete based on its weight. If we delete $e_{v(f)}$, then we no longer need to consider any other bad edge pairs involving $e_{v(f)}$, so we update $f \leftarrow f+1$ and consider the next edge in the list. Otherwise, we delete edge $e_{v(b)}$ and decrement the back pointer: $b \leftarrow b - 1$. If at any point we encounter an edge that was added to $\mathcal{D}$ previously, we move past it by incrementing $f$ or decrementing $b$. We stop this procedure when $\ell(v(f)) = \ell(v(b))$. Even if $f \neq b$, our assumption that edges are ordered by color means that there are no more bad edge pairs containing $v$ to consider. Figure~\ref{fig:pittstep} is an illustration of this process. 
Since we update either $f$ or $b$ in each step, and each step involves $O(1)$ operations, covering all bad edge pairs containing $v$ takes $O(d_v)$ time. We refer to this algorithm as \textsf{PittColoring} (see Appendix~\ref{app:vc} for pseudocode), and end with a summarizing theorem.
\begin{theorem}
	\textsf{PittColoring} is a randomized 2-approximation algorithm for weighted \minecc{} that runs in $O(\sum_{v \in V} d_v) = O(\sum_{e \in E} |e|)$ time.
\end{theorem}
\begin{figure}[t]
	\centering
	\includegraphics[width = \linewidth]{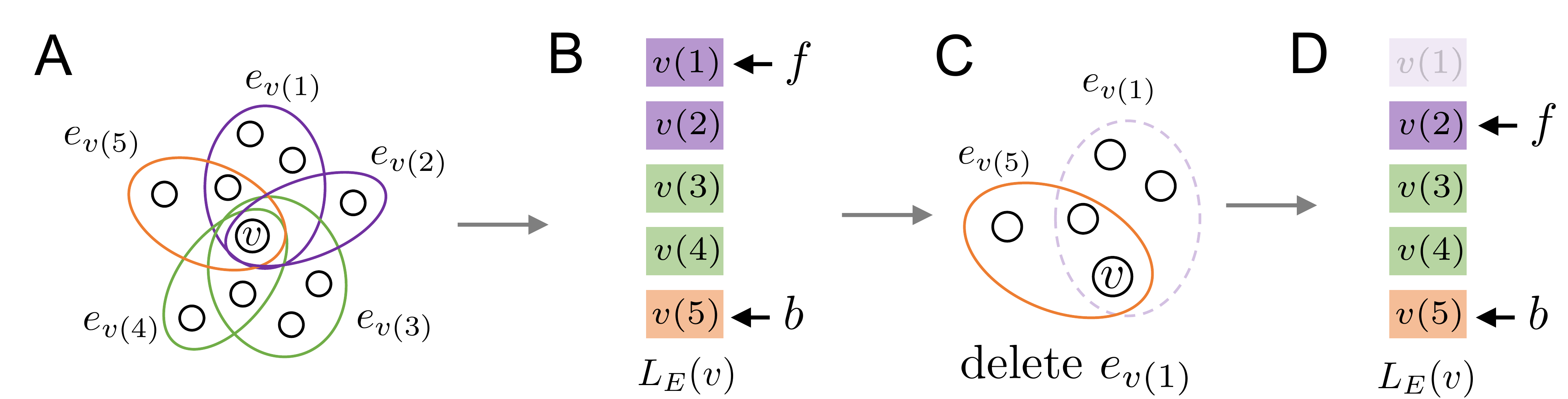}
	\vspace{-25pt}
	\caption{An illustration for covering bad edge pairs containing a node $v$. (A) $v$ is contained in 5 edges and 8 bad edge pairs. (B) Consider edges at opposite ends of $v$'s edge list $L_E(v)$, ordered by color. Edges $e_{v(1)}$ and $e_{v(5)}$ define a bad edge pair. (C) Given a bad edge pair, sample one to delete. (D) Once one $e_{v(1)}$ is deleted, we no longer consider bad edge pairs containing this edge, so we can advance the pointer $f$ to the next edge in $L_E(v)$.}
	\vspace{-5pt}
	\label{fig:pittstep}
\end{figure}
Our strategy for iterating through bad edge pairs can also be used in conjunction with other \textsc{Vertex Cover} algorithms. If edges are unweighted, every time we visit an uncovered bad edge pair in the list $L_E(v)$ we can instead delete \emph{both} edges and update both pointers $f$ and $b$. This leads to a new 2-approximation algorithm \textsf{MatchColoring} for the unweighted objective (see Appendix~\ref{app:vc}).
%This implicitly implements the standard strategy of finding a maximal matching and adding all nodes in the matching to a vertex cover. 
%This does not apply to edge-weighted \minecc{} but provides a deterministic 2-approximation for the unweighted version.

%
\section{Implementations and Experiments}
\label{sec:experiments}
Although our primary focus is to improve the theoretical foundations of edge-colored clustering, our algorithms are also easy to implement and very practical. 
%Source code and data for our experiments are available at \url{https://github.com/nveldt/ImprovedECC}.

\begin{table}[t]
	\caption{Approximation factors (ratio between algorithm output and LP lower bound) and runtimes obtained on five benchmark edge-colored hypergraphs. \textsf{MajorityVote} (\textsf{MV}) is deterministic. Our vertex cover algorithms \textsf{PittColoring} and \textsf{MatchColoring} both involve some amount of randomization, so we list the mean values and standard deviations over 50 runs.}
	\label{tab:approxmain}
	\centering
	\scalebox{0.95}{\begin{tabular}{l     l l l l  }
			\toprule
			& \multicolumn{4}{c}{\textbf{Ratio to LP lower bound}}  \\ 
			\cmidrule(lr){2-5} \
			\emph{Dataset} & {\small	\textsf{LP}}& {\small	\textsf{MV}} & {\small	\textsf{PittColoring}} & {\small	\textsf{MatchColoring}} \\
			\midrule
			{\small \texttt{Brain}} & 1.0 & 1.01 & 1.07 {\small $\pm 0.01$} & 1.08 {\small $\pm 0.01$} \\
			{\small \texttt{Cooking}} & 1.0 & 1.21 & 1.23 {\small $\pm 0.01$} & 1.23 {\small $\pm 0.0$}\\
			{\small \texttt{DAWN}} & 1.0 & 1.09 & 1.57 {\small $\pm 0.04$} & 1.58 {\small $\pm 0.03$} \\
			{\small \texttt{MAG-10}} & 1.0 & 1.18 & 1.39 {\small $\pm 0.01$} & 1.49 {\small $\pm 0.0$} \\
			{\small \texttt{Walmart}} & 1.0 & 1.2 & 1.13 {\small $\pm 0.0$} & 1.18 {\small $\pm 0.0$} \\
			\toprule
%			& \multicolumn{4}{c}{\textbf{Edge Satisfaction}}  \\ 
%			\cmidrule(lr){2-5} \
%			\emph{Dataset} & LP & MV & Pitt & Match \\
%			\midrule
%			Brain & 0.64 & 0.64 & 0.62 {\small $\pm 0.0$} & 0.62 {\small $\pm 0.0$} \\
%			Cooking & 0.2 & 0.03 & 0.01 {\small $\pm 0.0$} & 0.01 {\small $\pm 0.0$}\\
%			DAWN & 0.53 & 0.48 & 0.26 {\small $\pm 0.02$} & 0.25 {\small $\pm 0.02$}\\
%			MAG-10 & 0.62 & 0.55 & 0.47 {\small $\pm 0.0$} & 0.44 {\small $\pm 0.0$}\\
%			Walmart & 0.24 & 0.09 & 0.14 {\small $\pm 0.0$} & 0.11 {\small $\pm 0.0$}\\
%			\toprule
			& \multicolumn{4}{c}{\textbf{Runtime (in seconds)}}  \\ 
			\cmidrule(lr){2-5} \
			\emph{Dataset} & {\small	\textsf{LP}}& {\small	\textsf{MV}} & {\small	\textsf{PittColoring}} & {\small	\textsf{MatchColoring}} \\
			\midrule
		{\small 	\texttt{Brain}}&  0.52 & 0.001 & 0.006 {\small $\pm 0.028$} & 0.002 {\small $\pm 0.001$} \\
		{\small 	\texttt{Cooking}}&  127 & 0.002 & 0.01 {\small $\pm 0.003$} & 0.008 {\small $\pm 0.007$} \\
		{\small 	\texttt{DAWN}}&  4.23 & 0.003 & 0.01 {\small $\pm 0.004$} & 0.005 {\small $\pm 0.003$} \\
		{\small 	\texttt{MAG-10}}&  17.0 & 0.012 & 0.04 {\small $\pm 0.006$} & 0.04 {\small $\pm 0.016$}\\
		{\small 	\texttt{Walmart}}&  321 & 0.07 & 0.053 {\small $\pm 0.007$} & 0.05 {\small $\pm 0.009$} \\
			\bottomrule
	\end{tabular}}
	\vspace{-10pt}
\end{table} 

\textbf{Results on benchmark datasets.} Table~\ref{tab:approxmain} reports algorithmic results on the benchmark edge-colored hypergraphs of~\citet{amburg2020clustering}. On all datasets but \texttt{Walmart}, the \minecc{} LP relaxation produces integral solutions, i.e., we find the optimal \minecc{} solution without rounding. Even for \texttt{Walmart}, the LP solution is nearly integral and a very simple rounding scheme (namely, for each $v \in V$, assign it color $i^* = \argmin_i x_v^i$), produces a solution that is within a factor 1.00003 of optimal. For this reason, while our new LP rounding techniques improve the best theoretical results, is it not meaningful to compare them empirically against previous rounding techniques. It is however very meaningful to compare our vertex cover based algorithms against alternative approaches. Table~\ref{tab:approxmain} shows that these algorithms are orders of magnitude faster than solving and rounding the LP relaxation, and still obtain good approximate solutions. They also have the same asymptotic runtime as \textsf{MajorityVote}, a fast previous combinatorial algorithm that comes with a much worse $r$-approximation guarantee. Statistics for these datasets, additional implementation details, and more experimental results are provided in the appendix.

\begin{table}[t]
	\caption{Results (averaged over 50 runs) on the \texttt{Trivago} hypergraph. For \textsf{PittColoring}, the \textit{Apx} column gives the theoretical {expected} approximation guarantee. For other algorithms, \textit{Apx} is an a posteriori guarantee, improving on the method's theoretical guarantee, obtained using a lower bound computed by the algorithm.}
	\label{tab:trivago}
	\centering
	\scalebox{0.95}{\begin{tabular}{l     l l l l l }
			\toprule
			\textit{Algorithm} & \textit{Mistakes} & \textit{Sat} & \textit{Apx} & \textit{Acc} & \textit{Run} \\
			\midrule
%			\textsf{PittColoring} & 0.69  & $2^*$  & 0.70  & 0.13  \\
%			\textsf{MatchColoring} & 0.69  & 1.74  & 0.70  & 0.13  \\
%			\textsf{MajorityVote} & 0.73  & 36.43  & 0.77  & {0.12}  \\
%			\textsf{Hybrid} & {0.74}  & {1.46}  & {0.78}  & 0.25  \\
		{\small	\textsf{PittColoring}} & 74624 & 0.70  & $2$ & 0.72  & 0.13  \\
		{\small	\textsf{MatchColoring}} & 77606 & 0.69  & 1.74  & 0.70  & 0.12  \\
		{\small	\textsf{MajorityVote}} & 67941 & 0.73  & 36.43  & 0.77  & 0.11  \\
		{\small	\textsf{Hybrid}} & 65192 & 0.74  & 1.46  & 0.78  & 0.23  \\
			{\small\textsf{LP}} & 58031 & 0.77 & 1.001 & 0.80&1549\\
			\bottomrule
	\end{tabular}}
%	\vspace{-15pt}
\end{table} 

\textbf{Results for new \texttt{Trivago} hypergraph.} Our empirical contributions include a new benchmark dataset for edge-colored clustering, and a very practical hybrid algorithm that combines the strengths of \textsf{MatchColoring} and \textsf{MajorityVote}. The dataset is a hypergraph with $207974$ nodes, $247362$ edges, rank $r = 85$, and $k = 55$ colors. Nodes correspond to vacation rentals on the booking website \texttt{Trivago}, and each edge corresponds to the set of rentals that a user clicks on during a single user browsing session. Edge colors correspond to countries where the browsing session happens (e.g., \emph{Trivago.com} is the USA platform and \emph{Trivago.es} is for Spain). Nodes also come with ground truth colors that define the country the vacation rental is located in. Because users often (even if not always) search for vacation rentals within their own country, edge colors and ground truth node colors tend to match. This therefore serves as a useful new benchmark dataset for the ECC objective, that is much larger than previous benchmarks.

Table~\ref{tab:trivago} shows results for each algorithm on the \texttt{Trivago} hypergraph, including the number of mistakes made, the proportion of satisfied edges (\emph{Sat}), the method's approximation guarantee (\textit{Apx}), the number of nodes labeled correctly with respect to the ground truth (\textit{Acc}, i.e., accuracy), and the runtime (\textit{Run}, in seconds). \textsf{PittColoring}, \textsf{MatchColoring} and \textsf{MajorityVote} produce good results very quickly. A new algorithm \textsf{Hybrid} obtains even better results by first applying the edge deletion step of \textsf{MatchColoring} and then using the \textsf{MajorityVote} assignment for nodes that are isolated after edge deletions. \textsf{MajorityVote} is only guaranteed to return an 85-approximation since $r = 85$, while \textsf{MatchColoring} and \textsf{Hybrid} have a deterministic 2-approximation guarantee. We obtain improved a posteriori approximation guarantees for these three algorithms by comparing the number of mistakes they make against a lower bound on the optimal solution that can be computed based on the output of each algorithm. \textsf{PittColoring} has an \emph{expected} 2-approximation guarantee, but does not produce an explicit lower bound so it does not come with improved a posteriori guarantees. 
% Among the fast combinatorial methods, \textsf{Hybrid} obtains the best results. 

Solving and rounding the LP relaxation produces a very good output on the \texttt{Trivago} dataset, with an a posteriori approximation guarantee of $1.0014$, an edge satisfaction of $0.765$, and an accuracy of $0.80$. However, it takes around 25 minutes to solve the LP relaxation, which is roughly four orders of magnitude slower than our linear-time combinatorial algorithms. This result does indicate that LP-based methods for \minecc{} produce great results when it is possible to run them. However, our combinatorial algorithms will be able to scale to extremely large datasets where it is infeasible to rely on LP-based techniques. More details for all of our experimental results are provided in Appendices~\ref{app:vc} and~\ref{sec:appexp}.

\section{Conclusion and Discussion}
We have presented improved algorithms and hardness results for Edge-Colored Clustering. Our combinatorial algorithms have asymptotically optimal runtimes, and if $k$ and $r$ are arbitrarily large, their 2-approximation guarantee is the best possible assuming the unique games conjecture. Our LP algorithms are also tight with respect to the integrality gap. One open question is to determine whether alternative techniques could lead to improved approximations for fixed values of $r$ or $k$. Previous work has shown that the optimal approximation factor for \textsc{Node-MC} and \textsc{Hyper-MC} coincides with the integrality gap of their LP relaxation, assuming the unique games conjecture~\cite{ene2013local}. An open question is whether a similar result holds for the \minecc{} LP. In other words, is it UGC-hard to approximate \minecc{} below the LP integrality gap?
We also proved that the \minecc{} LP relaxation is tighter than the \textsc{Node-MC} relaxation. An open question in this direction is to better understand the relationship between the \minecc{} LP and relaxations obtained by considering more general problems, such as the Lov\'{a}sz relaxation for submodular multiway partition~\cite{chekuri2011submodular} or the basic LP for minimum constraint satisfaction problems~\cite{ene2013local}.

\newpage
\bibliography{chromaticclustering,multiway-cut-bib,vertex-cover-bib}
\bibliographystyle{icml2022}

%%%%%%%%%%%%%%%%%%%%%%%%%%%%%%%%%%%%%%%%%%%%%%%%%%%%%%%%%%%%%%%%%%%%%%%%%%%%%%%
%%%%%%%%%%%%%%%%%%%%%%%%%%%%%%%%%%%%%%%%%%%%%%%%%%%%%%%%%%%%%%%%%%%%%%%%%%%%%%%
% APPENDIX
%%%%%%%%%%%%%%%%%%%%%%%%%%%%%%%%%%%%%%%%%%%%%%%%%%%%%%%%%%%%%%%%%%%%%%%%%%%%%%%
%%%%%%%%%%%%%%%%%%%%%%%%%%%%%%%%%%%%%%%%%%%%%%%%%%%%%%%%%%%%%%%%%%%%%%%%%%%%%%%
\newpage
\appendix
\onecolumn

\section{Extended Related Work}
\label{app:related}

\paragraph{Chromatic correlation clustering.} 
The graph version of ECC is closely related to chromatic correlation clustering~\cite{Bonchi2015ccc,Anava:2015:ITP:2736277.2741629,klodt2021color,xiu2022chromatic}, 
an edge-colored generalization of correlation clustering~\cite{BansalBlumChawla2004}. Chromatic correlation clustering 
also takes an edge-colored graph as input and seeks to cluster nodes in a way that minimizes edge mistakes. ECC and chromatic correlation clustering both include penalties for (1) separating two nodes that share an edge, and (2) placing an edge of one color in a cluster of a different color. The key difference is that chromatic correlation clustering also includes a penalty for placing two non-adjacent nodes in the same cluster. 
In other words, chromatic correlation clustering interprets a non-edge as an indication that two nodes are dissimilar and should not be clustered together, whereas \minecc{} treats non-edges simply as missing information and does not include this type of penalty. 
%This results from a difference of interpretation: if two nodes do not share an edge, chromatic correlation clustering interprets this as an indication that the two nodes are dissimilar and should not be clustered together, whereas \minecc{} essentially treats non-edges as missing information (e.g., the nodes could be related but an edge was simply not observed).
As a result, a solution to chromatic correlation clustering may involve multiple different clusters of the same color, rather than one cluster per color. Unlike \minecc{}, chromatic correlation clustering is known to be NP-hard even in the case of a single color. Various constant-factor approximation algorithms have been designed, culminating in a 2.5-approximation~\cite{xiu2022chromatic}, but these do not apply to the general weighted case. Another difference is that there are no results for chromatic correlation clustering in hypergraphs.

\paragraph{Multiway cut and partition problems.}
The reduction from \minecc{} to a special case of \textsc{Node-MC} situates \minecc{} within a broad class of multiway cut and multiway partition problems~\cite{ene2013local,garg2004multiway,Dahlhaus94thecomplexity,calinescu2000,chekuri2011approximation,chekuri2011submodular,chekuri2016simple}. For \textsc{Node-MC}~\cite{garg2004multiway}, one is given a node-weighted graph $G = (V,E)$ with $k$ special terminal nodes, and the goal is to remove a minimum weight set of nodes to disconnect all terminals. \textsc{Node-MC} is approximation equivalent to the hypergraph multiway cut problem (\textsc{Hyper-MC})~\cite{chekuri2011submodular}, where one is give a hypergraph $H = (V,E)$ and $k$ terminal nodes, and the goal is to remove the minimum weight set of edges to separate terminals. These problems generalize the standard edge-weighted multiway cut problem in graphs (\textsc{Graph-MC})~\cite{Dahlhaus94thecomplexity}, where the goal is to separate $k$ terminal nodes in a graph by removing a minimum weight set of edges. \textsc{Node-MC} and \textsc{Hyper-MC} are also special cases of the more general submodular multiway partition problem (\textsc{Submodular-MP}), which has a best known approximation factor of $2(1-1/k)$~\cite{ene2013local}.
This matches the best approximation guarantee for \textsc{Node-MC} but requires solving and rounding a generalized Lov\'{a}sz relaxation~\cite{chekuri2011submodular}. It is worth noting that when $k = 2$, \minecc{} can be reduced to the minimum $s$-$t$ cut problem, i.e., the 2-terminal version of \textsc{Graph-MC}, but this does not generalize to $k > 2$. ~\citet{amburg2020clustering} showed a way to approximate \minecc{} by approximating a related instance of \textsc{Graph-MC}, but the objectives differ by a factor $(r+1)/2$. 
\section{Proofs for Linear Programming Results}

\textbf{Proof of Lemma~\ref{lem:intgap} (LP integrality gap)}

\textit{Proof.}
	We will construct a hypergraph $H$ that has exactly one edge $e_c$ for each color $c \in \{1,2, \hdots, k\}$, and a total of ${k \choose 2}$ nodes, each of which participates in exactly two edges. More precisely, we index each node by a pair of color indices $(i,j) \in C \times C$ with $i \neq j$. We place node $(i,j)$ in the hyperedge $e_i$ and the hyperedge $e_j$. Each edge contains exactly $r = k-1$ nodes. The optimal ECC solution makes a mistake at all but one edge. To see why, if we do not make a mistake at $e_c$, this means every node in $e_c$ is given label $c$. For each color $i \neq c$, the node with index $(i,c)$ is in $e_c$ and $e_i$, which means we will make a mistake at $e_i$ since this node was given label $c$. Meanwhile, the optimal solution to the LP relaxation is $\frac{k}{2}$: for node $v$ corresponding to color pair $(i,j)$, we set $x_{v}^i = x_{v}^j = \frac12$ and $x_v^c = 1$ for every $c \notin \{i,j\}$. Thus, $x_e = \frac12$ for each hyperedge $e \in E$, for an overall LP objective of $\frac{k}{2}$. \hfill $\square$

\textbf{Proof of Theorem~\ref{thm:rthm} (Hypergraph approximation result)}

\textit{Proof.}
	Let $e$ be an arbitrary edge with color $c = \ell(e)$. We need to show that $\pr{\eim} \leq 2 \big(1 - \frac{1}{r+1}\big) x_e$. If $x_e \geq \frac23$, this is trivial since $\pr{\eim}  \leq 1 \leq \frac{3}{2} x_e \leq 2 \big(1 - \frac{1}{r+1}\big) x_e$ as long as $r \geq 3$. We cover two remaining cases.
	
	\textbf{Case 1: $x_e < \frac12$.} For every $\rho \in (\frac12, \frac23)$, color $c$ wants all nodes in $e$, and Observation~\ref{xw} implies that $z_1> 1/2$. Therefore, the algorithm can only make a mistake at $e$ if $\rho$ falls between $z_1$ and $\frac23$. This never happens if $z_1 \geq \frac23$, so assume that $z_1 < \frac23$. Because $\rho < \frac23$, for every $v \in e$ it is possible for two different colors to want $v$, but never three colors. Since color $c$ is guaranteed to want every $v \in e$, the worse case scenario is when $\rho > z_1$ and each $v \in e$ is wanted by a different color that is unique to that node. In this case, the probability of making a mistake at $e$ is at most $\frac{r}{r+1}$. Putting all of the pieces together we have
	\begin{align*}
		\pr{\eim} &=\pr {\rho > z_1} \prc{\eim}{\rho > z_1}  \\
		&= \frac{\frac23 - z_1}{\frac23 - \frac12} \cdot \frac{r}{r+1} = \frac{6r}{r+1} \left( 1 - z_1 - \frac13\right) \\
		&\leq \frac{6r}{r+1} \left( x_e - \frac23x_e\right) = 2\left(1 - \frac{1}{r+1}\right)x_e.
	\end{align*}
	
	\textbf{Case 2: $x_e \in \left[ \frac12, \frac23\right)$.}
	If $\rho \leq x_e$, we are guaranteed to make a mistake at $e$, but this happens with bounded probability. If $\rho > x_e$, we can argue as in Case 1 that the probability of making a mistake at $e$ is at most $\frac{r}{r+1}$. We therefore have:
	\begin{align*}
		\pr{e \in M} &=\pr {\rho \leq x_e}\prc{\eim}{\rho \leq x_e} \\
		& \;\;\; + \pr {\rho > x_e} \prc{\eim}{\rho > x_e}  \\
		&= \frac{x_e - \frac12}{\frac23 - \frac12} \cdot 1  + \frac{\frac23 - x_e}{\frac23 - \frac12} \cdot \frac{r}{r+1}  \\
		&= 6 \left( x_e - \frac12 + \frac{2r}{3r+3} - \frac{x_er}{r+1} \right) \\
		&= 6 \left(\frac{x_e}{r+1}  + \frac{r-3}{6r+6}  \right) \\
		&\leq 6 \left(\frac{3x_e}{3r+3}  + \frac{(r-3)x_e}{3r+3}  \right) = 2 \left(1 - \frac{1}{r+1}\right)x_e.
	\end{align*}
\hfill $\square$

\textbf{Proof of Lemma~\ref{lem:xw} (Color Threshold Lemma)}

\textit{Proof.}
	Let $m$ be an arbitrary odd integer less than $k$. When $\rho > z_m$, the definition of $z_m$ tells us that there are $m$ colors (which are  distinct from each other and not equal to $c = \ell(e)$) that want at least one of the nodes in $e = (u,v)$. By the pigeonhole principle, at least $h = \ceil{\frac m2} = \frac{m+1}{2}$ of these colors want the \emph{same} node. Without loss of generality, let $u$ be this node and $\{c_1, c_2, \hdots, c_{h}\}$ be these ${h}$ colors, arranged so that
	%	, defined so that
	\begin{align*}
		x_u^{c_1} \leq x_u^{c_2} \leq \cdots \leq x_u^{c_{h}} \leq z_m.
	\end{align*}
	Without loss of generality we can assume these $h$ colors are the colors (not including $c$) that want node $u$ the most, in the sense that $x_u^j \geq x_u^{c_h}$ for every color $j \notin \{c, c_1, c_2, \hdots, c_h\}$.
	%Observe that $x_u^{c_{\bar{h}}} \leq z_m$; if $z_m < x_u^{c_{\bar{h}}}$, then for $\rho \in (z_m , x_u^{c_{\bar{h}}})$ we know that $\rho > z_m$ means there is a set of $\bar{h}$ colors that want $u$ by our assumption, but then the fact that $\rho < x_u^{c_{\bar{h}}}$ indicates that there are not that many nodes that want $u$. This would be a contradiction.
	Meanwhile, for every $\rho \leq z_m$, there can be at most $m - h = \frac{m-1}{2}$ \emph{other} colors not in $\{c, c_1, c_2, \hdots, c_h\}$ that want node $v$. We can show then that for $j \in \{1, 2, \hdots , h\}$, 
	\begin{equation}
		\label{eq:xucj}
		x_u^{c_j} \leq z_{j+m-h}.
	\end{equation}
	If this were not true and we instead assume $x_u^{c_j} > z_{j+m-h}$, this would mean that for any $\rho \in (z_{j + m - h}, x_u^{c_j})$, there are $j + m-h$ distinct colors (not equal to $c$) that want at least one of the two nodes in $e = (u,v)$. However, since $\rho < x_u^{c_j}$, we know that $\{c_1, c_2, \hdots, c_{j-1}\}$ are the only colors (not counting $c$) that want node $u$. Furthermore, since $\rho < z_m$, there are at most $m - h$ \emph{distinct} colors (again not counting $c$) that want node $v$. Thus, $\rho < x_u^{c_j}$ would imply there are at most $j+m-h -1$ colors not equal to $c$ that want either $u$ or $v$. This is a contradiction, so~\eqref{eq:xucj} must hold. Combining this inequality with the fact that $\sum_{i = 1}^k x_u^i = k-1$, we obtain
	%\begin{align*}
	%	(1- x_u^c) + \sum_{i = 1}^h (1- x_u^{c_i}) &\leq 1 \\
	%	\implies h  &\leq x_u^c + \sum_{i = 1}^h x_u^{c_i} \\
	%	\implies h &\leq x_e + \sum_{i = 1}^h z_{i + m - h} = x_e + z_{1 + m - h} + z_{2 + m - h} + \cdots + z_m.
	%\end{align*}
	\begin{align*}
		h  \leq x_u^c + \sum_{j = 1}^h x_u^{c_j} \leq x_e + \sum_{j= 1}^h z_{j + m - h} = x_e + z_{1 + m - h} + z_{2 + m - h} + \cdots + z_m.
	\end{align*}
	Since we assumed that $m$ is odd, we know $h = 1+m-h$ and $m = 2h - 1$, so setting $t = h$ yields the inequality in the statement of the lemma. \hfill $\square$

\subsection*{Proofs for Lemmas on Auxiliary Linear Programs}
Lemmas~\ref{lem:lp12} and~\ref{lem:lp34} are concerned with the LPs in Figures~\ref{fig:lpA} and~\ref{fig:lpB}. For notational ease in proving these results, let $M$ be the event $\eim$, i.e., the event of making a mistake at edge $e$ when rounding the LP relaxation. These Lemmas allow us to bound $\frac{\pr{M}}{x_e}$ by solving a small linear program, first for the case $x_e \in (\frac18,\frac12)$, and then for the case $x_e \in \left[\frac12, \frac34\right)$. Both results will repeatedly apply the following facts:
\begin{align}
	\label{eq:ii1}
	&\prc{M}{\rho > x_e \text{ and } \rho \in (z_i, z_{i+1})} = \frac{i}{i+1},\\
	\label{eq:83}
	&\pr{\rho \in (a, b)} = \frac{8}{3}(b - a), \text{ for $(a,b) \subseteq \Big(\frac12, \frac78\Big)$}.
\end{align}
\begin{figure}
	\begin{equation}
%		\label{lpA}
%		\boxed{
			\begin{array}{lll}
				\maximize & {\displaystyle \frac{q}{q+1}\frac78\chi - \sum_{j = 1}^q \frac{1}{j(j+1)} \omega_j }\\\\
				\text{subject to} 
				&\textbf{ (A\{i\}) } \omega_{i} - \omega_{i+1} \leq 0 \text{ for $i = 1, \hdots, 5$ } \\
				&\textbf{ (A6) } \chi - \omega_1 \leq 1 \\
				&\textbf{ (A7) } 2\chi - \omega_2 - \omega_3 \leq 1 \\
				&\textbf{ (A8) }3 \chi - 3\omega_5\leq 1 \\
				&\textbf{ (A9) } -\chi \leq -2\\
				& \textbf{ (A10) } \omega_q - \frac78 \chi \leq 0. 
			\end{array}
%		}
	\end{equation}	
	\caption{This linear program provides an upper bound for $\pr{\eim}/x_e$ when $x_e < \frac12\leq z_1 \leq z_{q} \leq \frac78$ for integers $q \leq 6$. We explicitly name each constraint as it will be convenient to reference individual constraints in later parts of the proof.}	
	\label{fig:lpA}
\end{figure}

\textbf{Proof of Lemma~\ref{lem:lp12}}

\textit{Proof.}
	When $x_e < \frac12$, we know that for any $\rho \in (\frac12, \frac78)$, the color $c = \ell(e)$ will want both nodes in $e = (u,v)$. If we use Lemma~\ref{lem:xw} with $t = 4$ and the monotonicity of color thresholds from~\eqref{eq:wmonotone}, we can see that
	\begin{align*}
		4 \leq x_e +  z_4 + z_5 + z_6 + z_7 \leq x_e + 4z_7 \implies 1- \frac{x_e}{4} \leq z_7 \implies z_7 > \frac78.
	\end{align*}
	This implies that for a random $\rho \in (\frac12, \frac78)$, at most 6 colors other than $c$ will want a node from $e = (u,v)$. Thus, if $z_{p-1} \leq \frac12\leq z_p \leq z_{q} \leq \frac78 \leq z_{q+1}$, then $q \leq 6$. Similarly, we can see that $1 \leq x_e + z_1 \implies z_1 > \frac12$, so we must have $p = 1$. Next, we use Eq.~\eqref{eq:ii1} and Eq.~\eqref{eq:83} to provide a convenient expression for $\pr{M}$:
	\begin{align*}
		\pr{M} & = 
		%		\prc{M}{\rho \in (1/2, z_1)} \pr{\rho \in (1/2, z_1)} + 
		\sum_{j = 1}^{q-1} \prc{M}{\rho \in (z_j, z_{j+1})} \pr{\rho \in (z_j, z_{j+1})} 
		+  \prc{M}{\rho \in (z_q, 7/8)} \pr{\rho \in (z_q,7/8)} \\
		&= \frac83\left(\frac{q}{q+1}\left(\frac78 - z_q\right)  +\sum_{j = 1}^{q-1} \frac{j}{j+1}(z_{j+1} - z_j)    \right)=  \frac83\left(\frac{q}{q+1}\frac78  - \sum_{j = 1}^q \frac{1}{j(j+1)} z_j   \right).
	\end{align*}
	Our goal is to upper bound $\frac{\pr{M}}{x_e}$. % \leq \frac43$.
	%	bounding the following ratio above by $\frac{4}{3}$:
	%	\begin{align}
	%		\frac{\pr{M}}{x_e} =  \left(\frac{q}{q+1}\frac78 - \frac{p-1}{2p}\right) \frac{1}{x_e}  - \frac83\sum_{j = p}^q \frac{1}{j(j+1)} \frac{z_j}{x_e} .
	%	\end{align}
	To do so we have the following inequalities and case-specific assumptions at our disposal: (1) the monotonicity of color thresholds: $z_i \leq z_{i+1}$ for $i \in \{0,1, \hdots , k\}$, (2) the relationship between $x_e$ and color thresholds given in Lemma~\ref{lem:xw}, and (3) our assumptions that $x_e < \frac12$ and $ z_{q} \leq \frac78 \leq z_{q+1}$.
	Since at most 6 colors other than $c$ can want a node in $e = (u,v)$, we can extract a set of 10 convenient inequalities from these three categories that we know are true for $x_e$ and the first 6 color threshold values: % , and must hold for the case we are considering in this lemma:
	%	that apply to 7 values:
	\begin{align*}
		&\textbf{Monotonicity Constraints}\\
		&\hspace{.5cm}\textbf{ (A1) } \frac{z_1}{x_e} \leq \frac{z_2}{x_e}, \textbf{ (A2) } \frac{z_2}{x_e} \leq \frac{z_3}{x_e}, \textbf{ (A3) }\frac{z_3}{x_e} \leq \frac{z_4}{x_e}, \textbf{ (A4) } \frac{z_4}{x_e} \leq \frac{z_5}{x_e}, \textbf{ (A5) } \frac{z_5}{x_e} \leq \frac{z_6}{x_e} \\ 
		&\textbf{Edge-Node Relationship Constraints (Lemma~\ref{lem:xw})}\\
		&\hspace{.5cm}\textbf{ (A6) } \frac{1}{x_e} \leq 1 + \frac{z_1}{x_e}, \textbf{ (A7) } \frac{2}{x_e} \leq 1 + \frac{z_2}{x_e} + \frac{z_3}{x_e}, \textbf{ (A8) } \frac{3}{x_e} \leq 1 + \frac{3z_5}{x_e}  \\ 
		&\textbf{Boundary Assumption Constraints}\\
		&\hspace{.5cm}\textbf{ (A9) } 1\leq \frac12 \frac{1}{x_e},  \textbf{ (A10) } \frac{z_q}{x_e} \leq \frac78 \frac{1}{x_e}.  
	\end{align*}
	If the maximum value of $\frac{\pr{M}}{x_e}$ over all values of $\{x_e, z_1, z_2, z_3, z_4, z_5, z_6\}$ that respect these inequalities is at most $P$, then this guarantees $\pr{M} \leq Px_e$. Finally, note that $\frac{\pr{M}}{x_e}$ and the inequalities above can be given as linear expressions in terms of $\frac{1}{x_e}$ and $\big\{\frac{z_i}{x_e} \colon i \in \{1,2,3,4,5,6\}\big\}$. We can therefore maximize $\frac38\frac{\pr{M}}{x_e}$ subject to these inequalities by setting $\chi = \frac{1}{x_e}$ and $\omega_i = \frac{z_i}{x_e}$ and solving the linear program in Figure~\ref{fig:lpA}. Note that we choose to maximize $\frac38\frac{\pr{M}}{x_e}$, rather than simply maximizing $\frac{\pr{M}}{x_e}$, only because this will simplify some of our analysis later.
	% The factor $\frac38$ in this objective function is included We have multiplied the objective function by  so that the optimal solution is an upper bound for $\frac38\pr{M}x_e$. While this is not strictly necessary it simplifies some of our analysis later. 
\hfill $\square$

\begin{figure}
	%\boxed{
	\begin{equation}
%		\label{lpB}
%		\boxed{
			\begin{array}{lll}
				\text{max} & {\displaystyle \frac1p + \left(\frac{q}{q+1}\frac78 -\frac12\right) \chi - \sum_{j = p}^q \frac{1}{j(j+1)} \omega_j   }\\\\
				\text{s.t.} 
				& \textbf{ (B\{i\}) } \omega_{i} - \omega_{i+1} \leq 0  \text{ for $i = 1, \hdots, 9$ } \\
				&\textbf{ (B10) } \chi - \omega_1 \leq 1 \\
				&\textbf{ (B11) } 2\chi - \omega_2 - \omega_3 \leq 1 \\
				&\textbf{ (B12) }3 \chi - \omega_3 - \omega_4 - \omega_5\leq 1 \\
				& \textbf{ (B13) }4 \chi - 4\omega_7 \leq 1\\
				&\textbf{ (B14) } \omega_{p-1} \leq 1 \\
				& \textbf{ (B15) } -\omega_p \leq -1 \\
				& \textbf{ (B16) } \omega_q - \frac78 \chi\leq 0.
			\end{array}
%		}
	\end{equation}	
	\caption{This linear program provides an upper bound for $\pr{\eim}/x_e$ when $x_e \in [\frac12, \frac34)$ and $z_{p-1} \leq x_e\leq z_p \leq z_{q} \leq \frac78 \leq z_{q+1}$ for integers $p \leq 5$ and $q \leq 10$ satisfying  $p \leq q$. We explicitly name each constraint as it will be convenient to reference individual constraints in later parts of the proof. We could add additional constraints based on the assumptions in Lemma~\ref{lem:lp34} and the relationship between variables given in Lemma~\ref{lem:xw}, but it suffices to consider a subset of these constraints to bound $\pr{\eim}/x_e$, and this also simplifies the analysis.}	
	\label{fig:lpB}
\end{figure}

\textbf{Proof of Lemma~\ref{lem:lp34}}

\textit{Proof.}
	Combining Lemma~\ref{lem:xw} (with $t = 6$), the inequalities in~\eqref{eq:wmonotone}, and the bound $x_e < \frac34$, we have
	\begin{align*}
		6 \leq x_e +  \sum_{j = 6}^{11} z_j \leq x_e + 6z_{11} \implies 1- \frac{x_e}{6} \leq z_{11} \implies z_{11} > \frac78.
	\end{align*}
	Thus, for every $\rho \in (\frac12, \frac78)$ we know $q \leq 10$ and there are at most 10 colors other than $c = \ell(e)$ that want a node in $e$. Using similar steps we can show that $z_5 > \frac34 > x_e$, so $p \leq 5$. We again use Eq.~\eqref{eq:ii1} and Eq.~\eqref{eq:83} to derive the following useful expression for $\pr{M}$:
	\begin{align*}
		\pr{M} & = \prc{M}{\rho \in \Big(\frac12, x_e\Big)} \pr{\rho \in \Big(\frac12, x_e\Big)} + \prc{M}{\rho \in (x_e,z_p)} \pr{\rho \in (x_e,z_p)}  \\
		& \hspace{.5cm} + \left( \sum_{j = p}^{q-1} \prc{M}{\rho \in (z_j, z_{j+1})} \pr{\rho \in (z_j, z_{j+1})} \right) +  \prc{M}{\rho \in \Big(z_q, \frac78\Big)} \pr{\rho \in \Big(z_q,\frac78\Big)} \\
		&= \frac83\left(\left(x_e - \frac12\right) + \frac{p-1}{p}\left(z_p - x_e\right) + \left( \sum_{j = p}^{q-1} \frac{j}{j+1}(z_{j+1} - z_j) \right) + \frac{q}{q+1}\left(\frac78 - z_q\right)   \right)\\
		\implies & \frac{\pr{M}}{x_e} \leq \frac83\left(\frac1p + \left(\frac{q}{q+1}\frac78 -\frac12\right) \frac{1}{x_e} - \sum_{j = p}^q \frac{1}{j(j+1)} \frac{z_j}{x_e}   \right).
	\end{align*}
	We apply the monotonicity inequalities~\eqref{eq:wmonotone} for color thresholds, Lemma~\ref{lem:xw}, and our assumptions in the statement of the lemma to obtain a set of inequalities that can be used to bound $\frac{\pr{M}}{x_e}$.
	\begin{align*}
		&\textbf{ (B\{i\}) }\frac{z_i}{x_e} \leq \frac{z_{i+1}}{x_e} \text{ for $i = 1, 2, \hdots 9$}, \\
		&\textbf{ (B10) } \frac{1}{x_e} \leq 1 + \frac{z_1}{x_e}, \textbf{ (B11) } \frac{2}{x_e} \leq 1 + \frac{z_2}{x_e} + \frac{z_3}{x_e}, \textbf{ (B12) } \frac{3}{x_e} \leq 1 + \sum_{j = 3}^5 \frac{z_j}{x_e}, \\
		&\textbf{ (B13) } \frac{4}{x_e} \leq 1 + 4\frac{z_7}{x_e}, \textbf{ (B14) } \frac{ z_{p-1} }{x_e} \leq 1, \textbf{ (B15) } 1 \leq \frac{z_{p}}{x_e}, \textbf{ (B16) } \frac{z_q}{x_e} \leq \frac78 \frac{1}{x_e}.
	\end{align*}
	We can maximize $\frac38\frac{\pr{M}}{x_e}$ subject to these inequalities by solving the linear program in Figure~\ref{fig:lpB}. \hfill $\square$

\paragraph{Proofs not contained in this section of the appendix}
The proof of Theorem~\ref{thm:lps} is given in Appendix~\ref{app:multiway}, which provides additional details on the relationship between \minecc{} and related multiway cut objectives. Details for proving Lemma~\ref{lem:hardx}, and in particular dual LP solutions for 46 auxiliary linear programs needed to prove this result, are contained in Appendix~\ref{app:dual}.

\section{Optimal Dual Variables and LP bounds for $r = 2$}
\label{app:dual}
\begin{table}
	\caption{Optimal dual variables (one for each constraint in the primal LP) for the linear program in Figure~\ref{fig:lpA}, for each valid integer $q$. Multiplying dual variables by the left and right hand sides of each constraint and adding the left and right hand sides together proves the upper bound (last column) on the optimal solution to the primal LP. The bound is always $\frac12$ or smaller.}
	\label{tab:lp12}
	\centering
	\begin{tabular}{l  cccccccccc  c}
		& {\footnotesize \rotatebox[origin=l]{70}{$\omega_{1} - \omega_{2} \leq 0$}} &{\footnotesize \rotatebox[origin=l]{70}{$\omega_{2} - \omega_{3} \leq 0$}}  &{\footnotesize \rotatebox[origin=l]{70}{$\omega_{3} - \omega_{4} \leq 0$}} &{\footnotesize \rotatebox[origin=l]{70}{$\omega_{4} - \omega_{5} \leq 0$}} &{\footnotesize \rotatebox[origin=l]{70}{$\omega_{5} - \omega_{6} \leq 0$}} &{\footnotesize \rotatebox[origin=l]{70}{$\chi - \omega_1 \leq 1$}}&{\footnotesize 
			\rotatebox[origin=l]{70}{$2\chi - \omega_2 - \omega_3 \leq 1$}}&{\footnotesize \rotatebox[origin=l]{70}{$3 \chi - 3\omega_5\leq 1$}}&{\footnotesize 
			\rotatebox[origin=l]{70}{$-\chi \leq -2$}}&{\footnotesize 
			\rotatebox[origin=l]{70}{$\omega_q - \frac78 \chi \leq 0$}} & \\
	\end{tabular}
	\begin{tabular}{l | cccccccccc  |c}
		\toprule
		$q$ & (\textbf{A1}) & (\textbf{A2}) & (\textbf{A3}) & (\textbf{A4}) & (\textbf{A5}) & (\textbf{A6}) & (\textbf{A7}) & (\textbf{A8}) & (\textbf{A9}) & (\textbf{A10}) & Bound \\
		\midrule
		1 & 0 & 0 & 0 & 0 & 0 & $\frac{1}{2}$ & 0 & 0 & $\frac{1}{16}$ & 0& $\frac{3}{8}$\\
		\midrule
		2 & $\frac16$ & 0 & 0 & 0 & 0 & $\frac23$ & 0 & 0 & $\frac{1}{12}$ & 0& $\frac{1}{2}$ \\
		\midrule
		3 & 0 & 0 & 0 & 0 & 0 & $\frac{1}{2}$ & $\frac16$ & 0 & $\frac{5}{48}$ & $\frac{1}{12}$ & $\frac{11}{24}$\\
		\midrule
		4 & 0 & 0 & $\frac{1}{12}$ & 0 & 0 & $\frac{1}{2}$ & $\frac16$ & 0 & $\frac{5}{48}$ & $\frac{1}{30}$& $\frac{11}{24}$ \\
		\midrule
		5 & 0 & 0 & $\frac{1}{12}$ & $\frac{1}{30}$ & 0 & $\frac{1}{2}$ & $\frac16$ & 0 & $\frac{5}{48}$ & 0& $\frac{11}{24}$ \\
		\midrule
		6 & 0 & 0 & $\frac{1}{12}$ & $\frac{1}{30}$& $\frac{1}{42}$ & $\frac{1}{2}$ & $\frac16$ & $\frac{1}{126}$ & $\frac{3}{28}$ & 0& $\frac{29}{63}$  \\
		\bottomrule
	\end{tabular}
\end{table}

In Section~\ref{sec:caseanalysis} of the main text we proved that when $q = 1$, the solution to the linear program in Figure~\ref{fig:lpA} is bounded above by $\frac12$. We accomplished this by multiplying constraints in this LP by dual variables corresponding to these constraints. In Table~\ref{tab:lp12}, we present dual variables for all cases $q \in \{1,2,3,4,5,6\}$. For the sake of completeness, here we provide additional details for how to prove an upper bound on the LP solution for all of these values of $q$, without having to explicitly write down and sum a linear combination of constraints. The upper bound can be shown more succinctly by checking a few matrix-vector products, which is the approach we take here. We also provide dual variables for the LP in Figure~\ref{fig:lpB} for all valid choices of $p$ and $q$, which can be used to prove the desired upper bound on this LP in the same fashion.

%For the sake of completeness, we will prove the necessary bound for each case in a way that only requires us to check a few matrix and vector products.  We first describe our strategy abstractly, and then apply it to check dual variables listed in our tables. 

\subsection{Overview of LP bounding strategy}
The linear programs in Figures~\ref{fig:lpA} and~\ref{fig:lpB} can both be written in the following format:
\begin{equation}
	\label{eq:primallp}
	\begin{array}{rl}
		\max & \vc^T \vx \\
		\text{s.t. } & \mA \vx \leq \vb.
	\end{array}
\end{equation}
The dual of this LP is given by
\begin{equation}
	\label{eq:duallp}
	\begin{array}{rl}
		\min & \vb^T \vy \\
		\text{s.t. } & \mA^T \vy = \vc \\
		& \vy \geq 0.
	\end{array}
\end{equation}
We want to prove that $\vc^T \vx^* \leq \frac12$ where $\vx^*$ is the optimal solution to the primal LP~\eqref{eq:primallp}. By weak duality theory, it is sufficient to find a vector $\vy$ of dual variables satisfying the constraints in the dual LP~\eqref{eq:duallp}, such that $\vb^T \vy \leq \frac12$. Note that our goal is actually to prove that this can be done for multiple slight variations of the objective function vector $\vc$ and constraint matrix $\mA$. For all cases we have the same right hand side vector $\vb$. If $\{\vc_1, \vc_2, \hdots, \vc_j\}$ denotes the set of objective function vectors, $\{\mA_1, \mA_2, \hdots, \mA_j\}$ is the set of corresponding constraint matrices, and $\{\vy_1, \vy_2, \hdots, \vy_j\}$ is an appropriately chosen set of dual vectors, then for $i = 1,2, \hdots , j$, we need to perform one matrix-vector product and one vector inner product to confirm that 
\begin{align*}
	\mA_j^T \vy_j &= \vc_j \\
	\vb^T \vy_j &\leq \frac12.
\end{align*}

\subsection{Matrix computations for bounding LP in Figure~\ref{fig:lpA}}
For the linear program in Figure~\ref{fig:lpA}, the variable vector is $\vx^T = \begin{bmatrix} \omega_1 & \omega_2 & \omega_3 & \omega_4 & \omega_5 & \omega_6 & \chi \end{bmatrix}$ and the right hand size vector is
$
\vb^T = \begin{bmatrix}
	0 & 0 & 0 & 0 & 0 & 1 & 1 & 1 &-2 & 0
\end{bmatrix}.
$
The dual variables we will use are given in Table~\ref{tab:lp12}. For each $q \in \{1,2, \hdots 6\}$, let $\vy_q$ denote the vector of dual variables (which corresponds to the $q$th row in Table~\ref{tab:lp12}), and let $\mY = \begin{bmatrix} \vy_1 & \vy_2 & \cdots & \vy_6 \end{bmatrix}$, so
\renewcommand{\arraystretch}{1.25}
\begin{equation*}
	\mY = \begin{bmatrix}
		0 & 0 & 0 & 0 & 0 & \frac{1}{2} & 0 & 0 & \frac{1}{16} & 0 \\ \frac{1}{6} & 0 & 0 & 0 & 0 & \frac{2}{3} & 0 & 0 & \frac{1}{12} & 0 \\ 0 & 0 & 0 & 0 & 0 & \frac{1}{2} & \frac{1}{6} & 0 & \frac{5}{48} & \frac{1}{12} \\ 
		0 & 0 & \frac{1}{12} & 0 & 0 & \frac{1}{2} & \frac{1}{6} & 0 & \frac{5}{48} & \frac{1}{30} \\ 
		0 & 0 & \frac{1}{12} & \frac{1}{30} & 0 & \frac{1}{2} & \frac{1}{6} & 0 & \frac{5}{48} & 0 \\ 0 & 0 & \frac{1}{12} & \frac{1}{30} & \frac{1}{42} & \frac{1}{2} & \frac{1}{6} & \frac{1}{126} & \frac{3}{28} & 0 
	\end{bmatrix}.
\end{equation*}
In order to check that $\vb^T \vy_q \leq \frac12$ for each choice of $q$ we just need to compute $\mY \vb$, which is made easier by the fact that most entries of $\vb$ are zero. 
%This computation is what produces the upper bounds on the primal linear program solutions:
\begin{align*}
	\mY \vb =  \begin{bmatrix}
		\frac{1}{2} & 0 & 0 & \frac{1}{16}  \\ 
		\frac{2}{3} & 0 & 0 & \frac{1}{12}  \\ 
		\frac{1}{2} & \frac{1}{6} & 0 & \frac{5}{48} \\ 
		\frac{1}{2} & \frac{1}{6} & 0 & \frac{5}{48} \\ 
		\frac{1}{2} & \frac{1}{6} & 0 & \frac{5}{48} \\ 
		\frac{1}{2} & \frac{1}{6} & \frac{1}{126} & \frac{3}{28} 
	\end{bmatrix}
	\begin{bmatrix}
		1 \\ 1 \\ 1 \\  -2
	\end{bmatrix}
	= \begin{bmatrix}
		\frac38 \\ \frac12 \\ \frac{11}{24} \\ \frac{11}{24} \\ \frac{11}{24} \\ \frac{29}{63}
	\end{bmatrix}.
\end{align*}
This confirms that as long as these dual variables are feasible for the dual linear program, the primal LP is always bounded above by $\frac12$. So, we just need to confirm that $\mA_q^T \vy_q = \vc_q$ for each $q$. 

The constraint matrix $\mA_q$ is nearly the same for all choices of $q$. Only the last constraint $(\omega_q - \frac78\chi \leq 0$) is dependent on $q$. The first nine rows of the constraint matrix are always given by
\begin{align*}
	\mA_q(1\colon9,\colon) &= \begin{bmatrix}
		1 &  -1 &  0 &  0 &  0 &  0 &  0 \\ 
		0 &  1 &  -1 &  0 &  0 &  0 &  0 \\ 
		0 &  0 &  1 &  -1 &  0 &  0 &  0 \\ 
		0 &  0 &  0 &  1 &  -1 &  0 &  0 \\ 
		0 &  0 &  0 &  0 &  1 &  -1 &  0 \\ 
		-1 &  0 &  0 &  0 &  0 &  0 &  1 \\ 
		0 &  -1 &  -1 &  0 &  0 &  0 &  2 \\ 
		0 &  0 &  0 &  0 &  -3 &  0 &  3 \\ 
		0 &  0 &  0 &  0 &  0 &  0 &  -1  
	\end{bmatrix}.
\end{align*}
The last row has two entries: $\mA_q(10,7) = -\frac78$ and $\mA_q(10,q) = 1$. The objective function is
\begin{align*}
	\vc_q^T \vx = \frac{7q}{8(q+1)} \chi - \sum_{j = 1}^q \frac{1}{j(j+1)}\omega_j,
\end{align*}
where $\vc_q$ is the objective function vector for $q \in \{1,2,\hdots, 6\}$. The matrix $\mC = \begin{bmatrix} \vc_1 & \vc_2 & \cdots & \vc_6 \end{bmatrix}$ of objective function vectors is given by:
\begin{equation*}
	\mC = \begin{bmatrix}
		-\frac{1}{2} &-\frac{1}{2} &-\frac{1}{2} &-\frac{1}{2} &-\frac{1}{2} &-\frac{1}{2} \\
		0 & -\frac{1}{6} &-\frac{1}{6} &-\frac{1}{6} &-\frac{1}{6} &-\frac{1}{6} \\
		0 & 0 & -\frac{1}{12} &-\frac{1}{12} &-\frac{1}{12} &-\frac{1}{12} \\
		0 & 0 & 0 & -\frac{1}{20} &-\frac{1}{20} &-\frac{1}{20} \\
		0 & 0 & 0 & 0 & -\frac{1}{30} &-\frac{1}{30} \\
		0 & 0 & 0 & 0 & 0 & -\frac{1}{42} \\
		\frac{7}{16} & \frac{7}{12} & \frac{21}{32} & \frac{7}{10} & \frac{35}{48} & \frac{3}{4} &  \\
	\end{bmatrix}
\end{equation*}
Given the matrices and vectors listed above, a few simple computations confirm that $\mA_q^T \vy_q = \vc_q$ indeed holds for every $q \in \{1,2, \hdots 6\}$.

\subsection{Dual variables for LP in Figure~\ref{fig:lpB}}
Tables~\ref{tab:lp34-1b} and~\ref{tab:lp34-2b} give optimal dual variables for the linear program in Figure~\ref{fig:lpB} for all possible values of $p$ and $q$. We also list the optimal solution to each linear program for each of the 40 cases, which is always less than or equal to $\frac12$. We highlight again that our proof does not rely on proving that this value is optimal---it suffices to use the dual variables to get a linear combination of inequalities such that summing the left hand sides produces the LP objective function, and summing the right hand sides yields the LP upper bound given in the last column of each row. We can use the strategy outlined in the previous subsection to check these dual variables and prove the upper bound by computing a matrix-vector product and a vector inner product for each case. We omit detailed calculations as they are straightforward but would take up a significant amount of space.

\begin{table*}
	\caption{Optimal dual variables, one for each of the 16 constraints, for the linear program in Figure~\ref{fig:lpB}, for each valid pair $(p,q)$ where $p \in \{1,2\}$. We omit the column for constraint number 6, $\omega_{6} - \omega_{7} \leq 0$, since the dual variable for this constraint is always zero for these cases. The last column is the bound on the LP objective (it is in fact the optimal LP solution value) that we can prove for each case using the dual variables. It is always $\frac12$ or less. }
	\label{tab:lp34-1b}
	\centering
	\begin{tabular}{r  | ccccccccccccccc |c }
		\toprule
		\rotatebox[origin=l]{90}{Constraints}& 
		{\footnotesize \rotatebox[origin=l]{90}{$\omega_{1} - \omega_{2} \leq 0$}} &
		{\footnotesize \rotatebox[origin=l]{90}{$\omega_{2} - \omega_{3} \leq 0$}} &
		{\footnotesize \rotatebox[origin=l]{90}{$\omega_{3} - \omega_{4} \leq 0$}} &
		{\footnotesize \rotatebox[origin=l]{90}{$\omega_{4} - \omega_{5} \leq 0$}} &
		{\footnotesize \rotatebox[origin=l]{90}{$\omega_{5} - \omega_{6} \leq 0$}} &
		%		{\footnotesize \rotatebox[origin=l]{90}{$\omega_{6} - \omega_{7} \leq 0$}} &
		{\footnotesize \rotatebox[origin=l]{90}{$\omega_{7} - \omega_{8} \leq 0$}} &
		{\footnotesize \rotatebox[origin=l]{90}{$\omega_{8} - \omega_{9} \leq 0$}} &
		{\footnotesize \rotatebox[origin=l]{90}{$\omega_{9} - \omega_{10} \leq 0$}} &
		{\footnotesize \rotatebox[origin=l]{90}{$\chi - \omega_1 \leq 1$}}&
		{\footnotesize \rotatebox[origin=l]{90}{$2\chi - \omega_2 - \omega_3 \leq 1$}}&
		{\footnotesize \rotatebox[origin=l]{90}{$3 \chi - \omega_3 - \omega_4 - \omega_5\leq 1$}}&
		{\footnotesize \rotatebox[origin=l]{90}{$4 \chi - 4\omega_7 \leq 1$}}&
		{\footnotesize \rotatebox[origin=l]{90}{$\omega_{p-1} \leq 1$}}&
		{\footnotesize \rotatebox[origin=l]{90}{$-\omega_{p} \leq -1$}}&
		{\footnotesize \rotatebox[origin=l]{90}{$\omega_q - \frac78 \chi \leq 0$}}& 
		{ \rotatebox[origin=l]{90}{LP bound}} \\
		\midrule
		$p$,$q$ & {\small$\textbf{1}$ } & {\small$\textbf{2}$ } & {\small$\textbf{3}$ } & {\small$\textbf{4}$ } & {\small$\textbf{5}$ } & {\small$\textbf{7}$ } & {\small$\textbf{8}$ } & {\small$\textbf{9}$ } & {\small$\textbf{10}$ } & {\small$\textbf{11}$ } & {\small$\textbf{12}$ } & {\small$\textbf{13}$ } & {\small$\textbf{14}$ } & {\small$\textbf{15}$ } & {\small$\textbf{16}$ } & \\ \midrule
		1, 1  & 0 & 0 & 0 & 0 & 0 & 0 & 0 & 0 & 0 & 0 & 0 & 0 & 0  & $\frac{4}{7} $ & $\frac{1}{14} $ & $\frac{3}{7} $\\ \midrule
		2   & $\frac{1}{6} $& 0 & 0 & 0 & 0 & 0 & 0 & 0  & $\frac{1}{12} $& 0 & 0 & 0 & 0  & $\frac{7}{12} $& 0  & $\frac{1}{2} $\\ \midrule
		3   & $\frac{3}{32} $ & $\frac{1}{192} $& 0 & 0 & 0 & 0 & 0 & 0 & 0  & $\frac{5}{64} $& 0 & 0 & 0  & $\frac{19}{32} $& 0  & $\frac{31}{64} $\\ \midrule
		4   & $\frac{1}{10} $ & $\frac{1}{30} $ & $\frac{1}{20} $& 0 & 0 & 0 & 0 & 0 & 0  & $\frac{1}{10} $& 0 & 0 & 0  & $\frac{3}{5} $& 0  & $\frac{1}{2} $\\ \midrule
		5   & $\frac{47}{450} $& 0  & $\frac{13}{900} $& 0 & 0 & 0 & 0 & 0 & 0  & $\frac{14}{225} $ & $\frac{8}{225} $& 0 & 0  & $\frac{136}{225} $ & $\frac{1}{450} $ & $\frac{37}{75} $\\ \midrule
		6   & $\frac{3}{28} $& 0  & $\frac{5}{252} $ & $\frac{17}{1260} $ & $\frac{1}{42} $& 0 & 0 & 0 & 0  & $\frac{5}{84} $ & $\frac{11}{252} $& 0 & 0  & $\frac{17}{28} $& 0  & $\frac{125}{252} $\\ \midrule
		7   & $\frac{7}{64} $& 0  & $\frac{37}{2016} $ & $\frac{257}{20160} $ & $\frac{1}{42} $& 0 & 0 & 0 & 0  & $\frac{11}{192} $ & $\frac{179}{4032} $ & $\frac{1}{224} $& 0  & $\frac{39}{64} $& 0  & $\frac{2003}{4032} $\\ \midrule
		8   & $\frac{1}{9} $& 0  & $\frac{13}{756} $ & $\frac{23}{1890} $ & $\frac{1}{42} $ & $\frac{1}{72} $& 0 & 0 & 0  & $\frac{1}{18} $ & $\frac{17}{378} $ & $\frac{1}{126} $& 0  & $\frac{11}{18} $& 0  & $\frac{94}{189} $\\ \midrule
		9   & $\frac{9}{80} $& 0  & $\frac{41}{2520} $ & $\frac{59}{5040} $ & $\frac{1}{42} $ & $\frac{1}{40} $ & $\frac{1}{90} $& 0 & 0  & $\frac{13}{240} $ & $\frac{229}{5040} $ & $\frac{3}{280} $& 0  & $\frac{49}{80} $& 0  & $\frac{2509}{5040} $\\ \midrule
		10   & $\frac{5}{44} $& 0  & $\frac{43}{2772} $ & $\frac{157}{13860} $ & $\frac{1}{42} $ & $\frac{3}{88} $ & $\frac{2}{99} $ & $\frac{1}{110} $& 0  & $\frac{7}{132} $ & $\frac{127}{2772} $ & $\frac{1}{77} $& 0  & $\frac{27}{44} $& 0  & $\frac{1381}{2772} $\\ \midrule
		2, 2  & 0 & 0 & 0 & 0 & 0 & 0 & 0 & 0  & $\frac{1}{12} $& 0 & 0 & 0  & $\frac{1}{12} $ & $\frac{1}{6} $& 0  & $\frac{1}{2} $\\ \midrule
		3  & 0  & $\frac{1}{192} $& 0 & 0 & 0 & 0 & 0 & 0 & 0  & $\frac{5}{64} $& 0 & 0 & 0  & $\frac{3}{32} $& 0  & $\frac{31}{64} $\\ \midrule
		4  & 0  & $\frac{1}{30} $ & $\frac{1}{20} $& 0 & 0 & 0 & 0 & 0 & 0  & $\frac{1}{10} $& 0 & 0 & 0  & $\frac{1}{10} $& 0  & $\frac{1}{2} $\\ \midrule
		5  & 0 & 0  & $\frac{13}{900} $& 0 & 0 & 0 & 0 & 0 & 0  & $\frac{14}{225} $ & $\frac{8}{225} $& 0 & 0  & $\frac{47}{450} $ & $\frac{1}{450} $ & $\frac{37}{75} $\\ \midrule
		6  & 0 & 0  & $\frac{5}{252} $ & $\frac{17}{1260} $ & $\frac{1}{42} $& 0 & 0 & 0 & 0  & $\frac{5}{84} $ & $\frac{11}{252} $& 0 & 0  & $\frac{3}{28} $& 0  & $\frac{125}{252} $\\ \midrule
		7  & 0 & 0  & $\frac{37}{2016} $ & $\frac{257}{20160} $ & $\frac{1}{42} $& 0 & 0 & 0 & 0  & $\frac{11}{192} $ & $\frac{179}{4032} $ & $\frac{1}{224} $& 0  & $\frac{7}{64} $& 0  & $\frac{2003}{4032} $\\ \midrule
		8  & 0 & 0  & $\frac{13}{756} $ & $\frac{23}{1890} $ & $\frac{1}{42} $ & $\frac{1}{72} $& 0 & 0 & 0  & $\frac{1}{18} $ & $\frac{17}{378} $ & $\frac{1}{126} $& 0  & $\frac{1}{9} $& 0  & $\frac{94}{189} $\\ \midrule
		9  & 0 & 0  & $\frac{41}{2520} $ & $\frac{59}{5040} $ & $\frac{1}{42} $ & $\frac{1}{40} $ & $\frac{1}{90} $& 0 & 0  & $\frac{13}{240} $ & $\frac{229}{5040} $ & $\frac{3}{280} $& 0  & $\frac{9}{80} $& 0  & $\frac{2509}{5040} $\\ \midrule
		10  & 0 & 0  & $\frac{43}{2772} $ & $\frac{157}{13860} $ & $\frac{1}{42} $ & $\frac{3}{88} $ & $\frac{2}{99} $ & $\frac{1}{110} $& 0  & $\frac{7}{132} $ & $\frac{127}{2772} $ & $\frac{1}{77} $& 0  & $\frac{5}{44} $& 0  & $\frac{1381}{2772} $ \\
		\bottomrule
	\end{tabular}
\end{table*}

\begin{table*}
	\caption{Optimal dual variables, one for each of the 16 constraints, for the linear program in Figure~\ref{fig:lpB}, for each pair $(p,q)$ when $p \in \{3,4,5\}$. We omit columns for constraint 1 ($\omega_{1} - \omega_{2} \leq 0$) and constraint 10 ($\chi - \omega_1 \leq 1$), as dual variables for these are always zero for these cases. }
	\label{tab:lp34-2b}
	\centering
	\begin{tabular}{ r | cccccccccccccc |c }
		\toprule
		\rotatebox[origin=l]{90}{Constraints}& 
		%		{\footnotesize \rotatebox[origin=l]{90}{$\omega_{1} - \omega_{2} \leq 0$}} &
		{\footnotesize \rotatebox[origin=l]{90}{$\omega_{2} - \omega_{3} \leq 0$}} &
		{\footnotesize \rotatebox[origin=l]{90}{$\omega_{3} - \omega_{4} \leq 0$}} &
		{\footnotesize \rotatebox[origin=l]{90}{$\omega_{4} - \omega_{5} \leq 0$}} &
		{\footnotesize \rotatebox[origin=l]{90}{$\omega_{5} - \omega_{6} \leq 0$}} &
		{\footnotesize \rotatebox[origin=l]{90}{$\omega_{6} - \omega_{7} \leq 0$}} &
		{\footnotesize \rotatebox[origin=l]{90}{$\omega_{7} - \omega_{8} \leq 0$}} &
		{\footnotesize \rotatebox[origin=l]{90}{$\omega_{8} - \omega_{9} \leq 0$}} &
		{\footnotesize \rotatebox[origin=l]{90}{$\omega_{9} - \omega_{10} \leq 0$}} &
		%		{\footnotesize \rotatebox[origin=l]{90}{$\chi - \omega_1 \leq 1$}}&
		{\footnotesize \rotatebox[origin=l]{90}{$2\chi - \omega_2 - \omega_3 \leq 1$}}&
		{\footnotesize \rotatebox[origin=l]{90}{$3 \chi - \omega_3 - \omega_4 - \omega_5\leq 1$}}&
		{\footnotesize \rotatebox[origin=l]{90}{$4 \chi - 4\omega_7 \leq 1$}}&
		{\footnotesize \rotatebox[origin=l]{90}{$\omega_{p-1} \leq 1$}}&
		{\footnotesize \rotatebox[origin=l]{90}{$-\omega_{p} \leq -1$}}&
		{\footnotesize \rotatebox[origin=l]{90}{$\omega_q - \frac78 \chi \leq 0$}} &
		{ \rotatebox[origin=l]{90}{LP bound}}  \\
		\midrule
		$p$, $q$ & {\small$\textbf{2}$ } & {\small$\textbf{3}$ } & {\small$\textbf{4}$ } & {\small$\textbf{5}$ } & {\small$\textbf{6}$ } & {\small$\textbf{7}$ } & {\small$\textbf{8}$ } & {\small$\textbf{9}$ } & {\small$\textbf{11}$ } & {\small$\textbf{12}$ } & {\small$\textbf{13}$ } & {\small$\textbf{14}$ } & {\small$\textbf{15}$ } & {\small$\textbf{16}$ } &\\ \midrule
		3, 3  & 0 & 0 & 0 & 0 & 0 & 0 & 0 & 0  & $\frac{5}{64} $& 0 & 0  & $\frac{5}{64} $ & $\frac{1}{192} $& 0  & $\frac{31}{64} $\\ \midrule
		4  & 0  & $\frac{1}{20} $& 0 & 0 & 0 & 0 & 0 & 0  & $\frac{1}{10} $& 0 & 0  & $\frac{1}{10} $ & $\frac{1}{30} $& 0  & $\frac{1}{2} $\\ \midrule
		5  & 0  & $\frac{13}{900} $& 0 & 0 & 0 & 0 & 0 & 0  & $\frac{14}{225} $ & $\frac{8}{225} $& 0  & $\frac{14}{225} $& 0  & $\frac{1}{450} $ & $\frac{37}{75} $\\ \midrule
		6  & 0  & $\frac{5}{252} $ & $\frac{17}{1260} $ & $\frac{1}{42} $& 0 & 0 & 0 & 0  & $\frac{5}{84} $ & $\frac{11}{252} $& 0  & $\frac{5}{84} $& 0 & 0  & $\frac{125}{252} $\\ \midrule
		7  & 0  & $\frac{37}{2016} $ & $\frac{257}{20160} $ & $\frac{1}{42} $& 0 & 0 & 0 & 0  & $\frac{11}{192} $ & $\frac{179}{4032} $ & $\frac{1}{224} $ & $\frac{11}{192} $& 0 & 0  & $\frac{2003}{4032} $\\ \midrule
		8  & 0  & $\frac{13}{756} $ & $\frac{23}{1890} $ & $\frac{1}{42} $& 0  & $\frac{1}{72} $& 0 & 0  & $\frac{1}{18} $ & $\frac{17}{378} $ & $\frac{1}{126} $ & $\frac{1}{18} $& 0 & 0  & $\frac{94}{189} $\\ \midrule
		9  & 0  & $\frac{41}{2520} $ & $\frac{59}{5040} $ & $\frac{1}{42} $& 0  & $\frac{1}{40} $ & $\frac{1}{90} $& 0  & $\frac{13}{240} $ & $\frac{229}{5040} $ & $\frac{3}{280} $ & $\frac{13}{240} $& 0 & 0  & $\frac{2509}{5040} $\\ \midrule
		10  & 0  & $\frac{43}{2772} $ & $\frac{157}{13860} $ & $\frac{1}{42} $& 0  & $\frac{3}{88} $ & $\frac{2}{99} $ & $\frac{1}{110} $ & $\frac{7}{132} $ & $\frac{127}{2772} $ & $\frac{1}{77} $ & $\frac{7}{132} $& 0 & 0  & $\frac{1381}{2772} $\\ \midrule
		4, 4   & $\frac{1}{10} $& 0 & 0 & 0 & 0 & 0 & 0 & 0  & $\frac{1}{10} $& 0 & 0  & $\frac{1}{5} $ & $\frac{1}{20} $& 0  & $\frac{1}{2} $\\ \midrule
		5   & $\frac{3}{64} $& 0 & 0 & 0 & 0 & 0 & 0 & 0  & $\frac{3}{64} $ & $\frac{1}{20} $& 0  & $\frac{23}{160} $& 0  & $\frac{1}{60} $ & $\frac{157}{320} $\\ \midrule
		6   & $\frac{5}{112} $& 0  & $\frac{1}{280} $ & $\frac{1}{42} $& 0 & 0 & 0 & 0  & $\frac{5}{112} $ & $\frac{3}{56} $& 0  & $\frac{1}{7} $& 0 & 0  & $\frac{55}{112} $\\ \midrule
		7   & $\frac{39}{896} $& 0  & $\frac{1}{280} $ & $\frac{1}{42} $& 0 & 0 & 0 & 0  & $\frac{39}{896} $ & $\frac{3}{56} $ & $\frac{1}{224} $ & $\frac{9}{64} $& 0 & 0  & $\frac{63}{128} $\\ \midrule
		8   & $\frac{43}{1008} $& 0  & $\frac{1}{280} $ & $\frac{1}{42} $& 0  & $\frac{1}{72} $& 0 & 0  & $\frac{43}{1008} $ & $\frac{3}{56} $ & $\frac{1}{126} $ & $\frac{5}{36} $& 0 & 0  & $\frac{71}{144} $\\ \midrule
		9   & $\frac{47}{1120} $& 0  & $\frac{1}{280} $ & $\frac{1}{42} $& 0  & $\frac{1}{40} $ & $\frac{1}{90} $& 0  & $\frac{47}{1120} $ & $\frac{3}{56} $ & $\frac{3}{280} $ & $\frac{11}{80} $& 0 & 0  & $\frac{79}{160} $\\ \midrule
		10   & $\frac{51}{1232} $& 0  & $\frac{1}{280} $ & $\frac{1}{42} $& 0  & $\frac{3}{88} $ & $\frac{2}{99} $ & $\frac{1}{110} $ & $\frac{51}{1232} $ & $\frac{3}{56} $ & $\frac{1}{77} $ & $\frac{3}{22} $& 0 & 0  & $\frac{87}{176} $\\ \midrule
		5, 5  & 0  & $\frac{8}{85} $& 0 & 0 & 0 & 0 & 0 & 0 & 0  & $\frac{8}{85} $& 0  & $\frac{16}{85} $& 0  & $\frac{31}{510} $ & $\frac{41}{85} $\\ \midrule
		6  & 0  & $\frac{8}{85} $& 0  & $\frac{31}{510} $& 0 & 0 & 0 & 0 & 0  & $\frac{8}{85} $& 0  & $\frac{16}{85} $& 0  & $\frac{22}{595} $ & $\frac{41}{85} $\\ \midrule
		7  & 0  & $\frac{8}{85} $& 0  & $\frac{31}{510} $ & $\frac{22}{595} $& 0 & 0 & 0 & 0  & $\frac{8}{85} $& 0  & $\frac{16}{85} $& 0  & $\frac{13}{680} $ & $\frac{41}{85} $\\ \midrule
		8  & 0  & $\frac{8}{85} $& 0  & $\frac{31}{510} $ & $\frac{22}{595} $ & $\frac{13}{680} $& 0 & 0 & 0  & $\frac{8}{85} $& 0  & $\frac{16}{85} $& 0  & $\frac{4}{765} $ & $\frac{41}{85} $\\ \midrule
		9  & 0  & $\frac{3}{32} $& 0  & $\frac{29}{480} $ & $\frac{41}{1120} $ & $\frac{1}{40} $ & $\frac{1}{90} $& 0 & 0  & $\frac{3}{32} $ & $\frac{1}{640} $ & $\frac{3}{16} $& 0 & 0  & $\frac{309}{640} $\\ \midrule
		10  & 0  & $\frac{41}{440} $& 0  & $\frac{79}{1320} $ & $\frac{111}{3080} $ & $\frac{3}{88} $ & $\frac{2}{99} $ & $\frac{1}{110} $& 0  & $\frac{41}{440} $ & $\frac{7}{1760} $ & $\frac{41}{220} $& 0 & 0  & $\frac{851}{1760} $ \\
		\bottomrule
	\end{tabular}
\end{table*}

\subsection{Challenges in Avoiding Case Analysis}
\label{app:badcases}
It is natural to wonder whether we can rule out or simultaneously handle a large number of cases for Lemmas~\ref{lem:lp12} and~\ref{lem:lp34}, rather than check a different set of dual variables for 46 small linear programs. However, there are several challenges in escaping from this lengthy case analysis. First of all, the 46 linear programs we have considered all have feasible solutions, so we cannot rule any out on the basis that they they correspond to impossible cases. The optimal solution for many cases is exactly $\frac{1}{2}$, so the inequalities we apply to bound $\frac{\pr{M}}{x_e}$ for these cases must be tight. For many of the other cases, the optimal solution is very close to $\frac12$. We attempted to use looser bounds to prove a $\frac{1}{2}$ bound for multiple cases at once, but were unsuccessful. Each case seems to require a slightly different argument in order to prove the needed bound. Adding more constraints to the linear programs also did not lead to a simpler analysis.

A second natural strategy is to try to first identify and prove which values of $x_e$, $p$, and $q$ correspond to the worst cases, and then simply confirm that the probability bound holds for these cases. However, the worst case values for $x_e$, $p$, and $q$ are not obvious, and are in fact somewhat counterintuitive. At first glance it may appear that the worse case is when $q$ is as large as possible and $p$ is as small as possible. This maximizes the expected number of colors that want a node in $e$ for a randomly chosen $\rho \in (\frac12, \frac78)$, leading to a higher probability of making a mistake at $e$. However it is important to recall that the goal is to bound $\frac{\pr{M}}{x_e}$, and not just $\pr{M}$, so this line of reasoning does not apply. As it turns out, the worst case scenario (i.e., largest optimal LP solution) for $x_e \in [\frac12, \frac34)$ is when $q \in \{2,4\}$, but we do not have a proof (other than by checking all cases) for why this case leads to the largest value of $\frac{\pr{M}}{x_e}$.

Finally, another natural question is whether we can avoid a lengthy case analysis by choosing an interval other than $I = (\frac12, \frac78)$ when applying Algorithm~\ref{alg:gen}. It is especially tempting to use an interval $(a,b)$ where $b < \frac78$, since this would potentially decrease the maximum number of colors that want a node in $e = (u,v)$. This is desirable because if fewer colors want a node in $e$, then we have fewer color threshold variables to worry about and fewer cases to consider. The following result (Lemma~\ref{lem:badcases} in the main text) shows why this approach will fail. 
\begin{lemma}
	\label{lem:badcases-restated}
	Let $Y$ denote the node color map returned by running Algorithm~\ref{alg:gen} with $I = (a,b) \subseteq (0,1)$ on an edge-colored graph.
	\begin{itemize}
		\item If $a > \frac12$, there exists a feasible LP solution such that $\pr{\eim} > \frac{4}{3} x_e$. 
		\item If $a \leq \frac12$ and $b < \frac78$, there exists a feasible LP solution such that $\pr{\eim} > \frac43 x_e$. 
	\end{itemize}
\end{lemma}
\begin{proof}
	If $a > \frac12$, then there exists some $\varepsilon >  0$ such that $a = \frac{1+\varepsilon}{2}$. Consider an edge $e = (u,v)$ with color $c = \ell(e)$ whose LP variables satisfy $x_e = x_u^c = x_v^c = \frac{1-\varepsilon}{2}$ and where $x_u^i = x_v^j = \frac{1+\varepsilon}{2} = a$ for two colors $i \neq j$ that are both distinct from $c$. This is feasible as long all as $x_u^t = 1$ for $t \notin \{c, i\}$ and $x_v^t = 1$ for $t \notin \{c,j\}$. For every $\rho \in (a,b)$, color $i$ will want node $u$, color $j$ will want node $v$, and color $c$ will want both of them. Therefore, the probability of making a mistake at $e$ is exactly
	\begin{align*}
		\pr{\eim} = \frac{2}{3} = \frac{2}{3} \cdot \frac{1}{x_e} x_e = \frac{2}{3} \cdot \frac{2}{1 - \varepsilon} x_e > \frac43 x_e.
	\end{align*}
	
	If instead $a \leq \frac12$, consider the feasible solution where $x_e = x_u^c = x_v^c = x_u^i = x_u^g = x_v^j = x_v^h = \frac23$ where $\{c,g,h,i,j\}$ are all distinct colors. If $b < \frac{2}{3}$ we are guaranteed to make a mistake at $x_e$, so assume that $\frac23 \leq b < \frac78$. Then we can calculate that
	\begin{align*}
		\frac{\pr{\eim}}{x_e} = \frac{1}{b-a} \left( \frac{2}{3} - a + \frac{4}{5} \left( b - \frac{2}{3} \right) \right)\cdot \frac32 = \frac{2 - 3a}{10(b-a)} + \frac{6}{5} > \frac{2 - 3 \cdot \frac12}{10(\frac78-\frac12)} + \frac{6}{5} = \frac43. 
	\end{align*}
\end{proof}
The first case in this lemma indicates that we should not use an interval $(a,b)$ where $a > \frac12$. The second case rules out the hope of choosing some $b < \frac78$ in order to make the analysis easier. The interval $I = (\frac12, \frac78)$ is chosen to avoid these two issues and be as simple to analyze as possible. 

Lemma~\ref{lem:badcases} and the other challenges highlighted above do not imply that it is impossible to design a $\frac{4}{3}$-approximation algorithm with a simpler approximation guarantee proof. Indeed, a simpler proof of a tight approximation is a compelling direction for future research. Nevertheless, this discussion highlights why it is challenging to escape from using a length case analysis when proving that $\pr{\eim} \leq \frac43 x_e$ for Algorithm~\ref{alg:gen}.
\section{Relation to Multiway Cut Objectives}
\label{app:multiway}
In this section we prove that the canonical \minecc{} LP relaxation is always at least as tight as the \textsc{Node-MC} LP relaxation, and that it can be strictly tighter in some cases. Therefore, although these approaches lead to a matching worst-case guarantee when $r$ is arbitrarily large, applying the canonical relaxation leads to better results in many cases. We will also highlight why our approximation guarantees for small values of $r$ are much better than any guarantees we obtain by reducing to \textsc{Hyper-MC}.

%In this section we expand on the known relationship between \minecc{} and generalized multiway cut objectives that have been studied extensively in theoretical computer science. In particular we will show that the canonical LP relaxation for \textsc{Color-LP} is always at least as right as the LP relaxation obtained by first reducing to an instance of \textsc{Node-MC}~\cite{garg2004multiway}, and that in some cases it is a strictly better lower bound. We will also rule out the possibility of obtaining improved approximation guarantees by rounding other types of relaxations that have been used for generalized multiway cut objectives~\cite{ene2013local,chekuri2011approximation,chekuri2011submodular}. 

\subsection{Reductions from \minecc{} to \textsc{Node-MC} and \textsc{Hyper-MC}}
\label{sec:reduced}
As shown previously~\cite{amburg2020clustering}, an instance $H = (V,E, C, \ell)$ of \minecc{} can be reduced to an instance of \textsc{Node-MC} on a graph $G$ via the following steps
\begin{itemize}
	\item For every node $v \in V$, add a corresponding node $v$ to $G$ with weight $\infty$.
	\item For each color $i \in \{1,2,\hdots , k\}$, place a terminal node $t_i$ in $G$.
	\item For each hyperedge $e \in E$ with weight $w_e$, add a node $v_e$ to $G$ with weight $w_e$ and place an edge $(v, v_e)$ for each $v \in e$.
\end{itemize}
The \minecc{} objective on $H$ is then approximation equivalent to \textsc{Node-MC} on $G$. \textsc{Node-MC} is in turn approximation equivalent to \textsc{Hyper-MC}~\cite{chekuri2011submodular}.
% implying the existence of a reduction from \minecc{} to \textsc{Hyper-MC}. 
Although not previously shown explicitly, there is a particularly simple reduction from \minecc{} to \textsc{Hyper-MC}: introduce a terminal node $t_i$ for each color $i$, and then for a hyperedge $e \in E$ of color $i$, add terminal node $t_i$ to that hyperedge and then ignore the hyperedge color.  Figure~\ref{fig:reductions} illustrates the procedure of reducing an instance of \minecc{} to \textsc{Node-MC} and \textsc{Hyper-MC}. We note the following simple equivalence result.
\begin{observation}
	\minecc{} with $k$ colors is approximation-equivalent to a special type of \textsc{Hyper-MC} problem on $k$ terminal nodes where every hyperedge contains a terminal node.
\end{observation}

\begin{figure}[t]
	\centering
	\includegraphics[width = .75\linewidth]{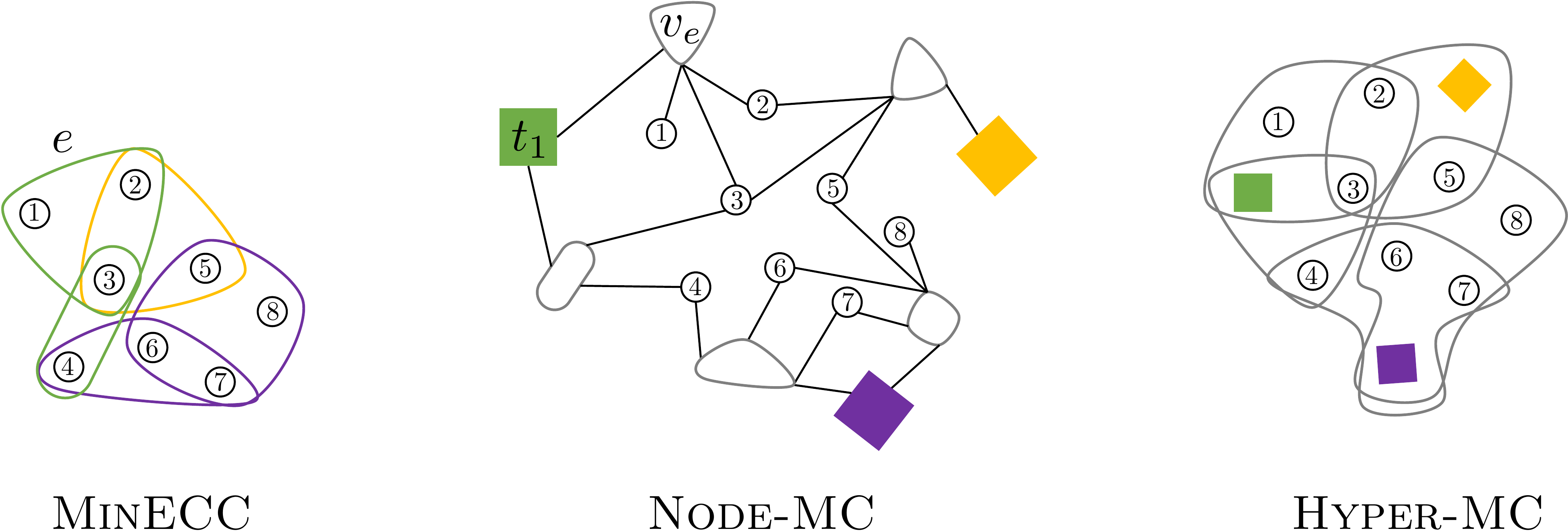}
	\caption{\minecc{} can be reduced to \textsc{Node-MC} and \textsc{Hyper-MC} in an approximation preserving way. Squares represent terminal nodes, one for each color. In the instance of \textsc{Node-MC}, original nodes (numbered circles) are given weight $\infty$, so we can only delete hyperedge nodes when separating terminal nodes. In the illustration, green is color 1, and we label one green edge $e = \{1,2,3\}$, which is converted to a new node $v_e$ that is attached to terminal $t_1$ in the \textsc{Node-MC} instance.}
	\label{fig:reductions}
\end{figure}
%Note that an instance of \minecc{} with maximum hyperedge size $r$ corresponds to an instance of \textsc{Hyper-MC} with maximum hyperedge size $\Delta = r+1$ so this does not imply a reduction from \minecc{} to the standard \textsc{Graph-MC} objective.

\paragraph{Approximation guarantees for \minecc{} via reductions}
There are a few known approximation guarantees for \textsc{Hyper-MC} in terms of the maximum hyperedge size $r$, but these do not imply any useful results for \minecc{}. When $r = 2$, \textsc{Hyper-MC} is the well-studied graph multiway cut objective~\cite{Dahlhaus94thecomplexity,calinescu2000}, but this has no direct bearing on~\minecc{} since an instance of \minecc{} with maximum hyperedge size $r \geq 2$ corresponds to an instance of \textsc{Hyper-MC} with maximum hyperedge size $r+1 \geq 3$. \citet{chekuri2011submodular} gave an $H_r$ approximation for \textsc{Hyper-MC} where $H_i$ is the $i$th harmonic number. However, this only implies a $1.833$-approximation for \minecc{} when $r = 3$, and is worse than a $2$-approximation when $r \geq 4$. The approximation guarantees we obtain for small values of $r$ by rounding the \minecc{} LP relaxation are therefore significantly better and than results obtained by reducing to generalized multiway cut objectives. In contrast, when $r$ is arbitrarily large, our approximation guarantee of $2(1- \frac1k)$ for rounding the \minecc{} LP relaxation matches the guarantee obtained by first reducing to \textsc{Node-MC} and rounding the LP relaxation for this objective~\cite{garg2004multiway}. In this appendix we explore the relationship between these linear programs in more depth.

\subsection{The \minecc{} LP is Tighter than the \textsc{Node-MC} LP}
Let $G = (V,E)$ be a graph with terminal nodes $T = \{t_1, t_2, \hdots, t_k\}$ and a weight $w_v \geq 0$ for each node $v \in V$. 
%Let $P_{ij}$ be the set of paths between nodes $t_i$ and $t_j$ in $G$, and $\mathcal{P} = \cup_{i,j} P_{ij}$ be the set of all paths between terminal nodes. 
The distance-based LP relaxation for the~\textsc{Node-MC} objective on $G$ is shown in Figure~\ref{fig:nwmclp}, where $\mathcal{P}$ represents the set of all paths between pairs of terminal nodes. The formulation in~\eqref{eq:nwmc} uses a single variable $d_v$ for each node $v$, with the constraint that the distance between every pair of terminal nodes is at least 1. This requires an exponential number of path constraints. The bottom of Figure~\ref{fig:nwmclp} shows an alternative way to write the LP using more variables but only polynomially many constraints and variables. The two LPs are known to be equivalent~\cite{garg2004multiway}. 
%To see why this is a relaxation, observe that if we add the constraint $d_v \in \{0,1\}$ to the path constraints version shown in~\eqref{eq:nwmc}, this exactly encodes the \textsc{Node-MC} problem: deleting a node $v$ corresponds to setting $d_v = 1$, and we must delete at least one node on each path between terminals. 
The variable $d_u$ provides an indication for how strongly we wish to delete node $u$, and the 
%In the alternate formulation with polynomially many constraints shown in~\eqref{eq:nwmc2}, the 
variable $y_u^i$ is interpreted as the distance between node $u$ and cluster~$i$.
% while variable $d_u$ indicates how strongly we wish to delete node $u$.
% smaller set of variables but an exponential number of {path} constraints. We focus on the formulation in Figure~\ref{fig:nwmclp} as it has polynomially many constraints. 
% To see why this is a relaxation for \textsc{Node-MC}, consider a set $S \subseteq V- T$ that separates all terminal pairs when deleted, and set $d_v = 1$ for each $v \in S$. Set $y_{u}^i = 0$ if there is a path between terminal $t_i$ and node $u$ after $S$ has been deleted, otherwise set $y_u^i = 1$. This yields a feasible solution for the linear program~\eqref{eq:nmclp} whose objective score equals the weight of deleted nodes. The variable $y_u^i$ is interpreted as the distance between node $u$ and cluster $i$, while variable $d_u$ indicates how strongly we wish to delete node $u$.
%\begin{figure}
%	\begin{equation}
%		\label{eq:nmclp}
%		\boxed{
%			\begin{array}{lll}
%				& \textbf{Node-MC-LP} & \\
%				\min & \displaystyle{\sum_{v \in V-T}} w_v d_v & \\ 
%				\text{s.t.} & y_v^i \leq y_u^i + d_v & \forall (u,v) \in E, \forall i \in \{1,2, \hdots, k\}\\\
%				& y_{t_i}^i = 0 & \forall i \in \{1,2, \hdots, k \}\\\
%				& y_{t_i}^j \geq 1 & \forall i, j \in \{1,2, \hdots, k\}, j \neq i\\ 
%				& d_t = 0 & \forall t \in T  \\
%				& d_v \geq 0 & \forall v \in V-T 
%			\end{array}
%		}
%	\end{equation}
%	\caption{The distance-based LP relaxation for node-weighted multiway cut~\cite{garg2004multiway}.}
%	\label{fig:nwmclp}
%\end{figure}

\begin{figure}
	%		\begin{minipage}{.5\linewidth}
	\textbf{Node-MC LP}: \text{path constraints version} 
	\begin{equation}
		\boxed{
			\label{eq:nwmc}
			\begin{array}{lll}
				%					& \textbf{Node-MC-LP}: \text{path constraints version} & \\
				\min & \displaystyle{\sum_{v \in V-T}} w_v d_v & \\ 
				\text{s.t.} & \sum_{v \in p} d_v \geq 1 & \forall  p \in \mathcal{P} \\
				& d_v = 0 & \forall v \in T \\
				& d_v \geq 0 & \forall v \in V-T 
			\end{array}
		}
	\end{equation}
	%		\end{minipage}
	%		\begin{minipage}{.5\linewidth}
	\textbf{Node-MC LP}: \text{polynomial constraints version} 
	\begin{equation}
		\label{eq:nwmc2}
		\boxed{
			\begin{array}{lll}
				%				& \textbf{Node-MC-LP}: \text{polynomial contraints version} & \\
				\min & \displaystyle{\sum_{v \in V-T}} w_v d_v & \\ 
				\text{s.t.} & y_v^i \leq y_u^i + d_v & \forall (u,v) \in E, \forall i \in [k]\\\
				& y_{t_i}^i = 0 & \forall i \in [k]\\\
				& y_{t_i}^j \geq 1 & \forall i, j \in [k], j \neq i\\ 
				& d_v \geq 0 & \forall v \in V-T 
			\end{array}
		}
	\end{equation}
	%	\end{minipage}
	\caption{Two equivalent ways to write the distance-based LP relaxation for node-weighted multiway cut~\cite{garg2004multiway}. $\mathcal{P}$ is the set of all paths between terminal nodes.}
	\label{fig:nwmclp}
\end{figure}

Consider now an instance of \minecc{} encoded by an edge-colored hypergraph $H = (V,E, C, \ell)$ that we reduce to a node-weighted graph $G = (\hat{V}, \hat{E})$ using the strategy in Section~\ref{sec:reduced}. The node set $\hat{V} = T \cup V \cup V_E$ is made up of terminal nodes $T = \{t_i \colon i = 1, 2, \hdots , k\}$, the original node set $V$ from hypergraph $H$, and the node set $V_E = \{v_e \colon e \in E\}$. In Figure~\ref{fig:reductions}, these are represented by square nodes, numbered circular nodes, and irregular shaped hyperedge-nodes, respectively. The edge set $\hat{E} = E_H \cup E_T$ has two parts, defined by
\begin{align}
	\label{eq:edgeset}
	E_H &= \{(v, v_e) \colon \text{$v \in e$ in $H$} \} \\
	E_T &= \{(t_i, v_e) \colon \text{$\ell(e) = i$ in $H$} \}.
\end{align}
The node weight for each $u \in V \cup T$ is given by $w_u = \infty$, and the node weight for $v_e \in V_E$ is $w_e$. Focusing specifically on the formulation shown in~\eqref{eq:nwmc2}, we see that the \textsc{Node-MC} LP shares some similarities with the canonical \minecc{} relaxation. For example, both linear programs involve one variable for each node-color pair $(u,i) \in V \times C$. However, in terms of the hypergraph $H = (V,E, C, \ell)$, the \textsc{Node-MC} relaxation involves $O(k|V| + k|E|)$ variables and $O(k\sum_{e \in E} |e|)$ constraints overall, while the \minecc{} relaxation has $O(k|V| + |E|)$ variables and $O(k|V| + \sum_{e \in E} |e|)$ constraints. Although the two linear programs both have an integrality gap of $2(1 - \frac{1}{k})$, the following theorem proves that the \minecc{} relaxation will always be at least as tight as the lower bound obtained via the \textsc{Node-MC} relaxation.
\begin{theorem}
	\label{thm:lps-app}
	(Theorem~\ref{thm:lps} in the main text)
	The value of the \textsc{Node-MC} relaxation on $G = (\hat{V}, \hat{E})$ is at most the value of the \minecc{} relaxation on $H = (V,E,C,\ell)$. 
\end{theorem}
\begin{proof}
	Let $\textbf{X} = \{x_v^i, x_e \colon v \in V, e \in E, i \in [k] \}$ denote an arbitrary set of variables for the \minecc{} LP relaxation for $H$. Our goal is to use these to construct a set of feasible variables for the \textsc{Node-MC} LP relaxation on graph $G = (\hat{V}, \hat{E})$ with the same objective score. 
	%	If this holds, then we know that the optimal solution the \textsc{Node-MC} LP relaxation is less than or equal to the optimal solution value
	This means we have to define $y_v^i$ and $d_v$ for each $v \in \hat{V}$ and $i \in [k]$. 
	Since there are multiple different types of nodes in $\hat{V}$, in order to simplify notation we will let $y_e^i = y_{v_e}^i$ and $d_e = d_{v_e}$ when we are considering a node $v_e \in V_E$ that is associated with a hyperedge $e \in E$. 
	%	We will also let $y_i^j = y_{t_i}^j$ 
	Define the following set of variables for the \textsc{Node-MC} LP:
	\begin{itemize}
		\item Let $d_v = 0$ for $v \in V \cup T$.
		\item For $v_e \in V_E$, define
		%		\begin{equation*}
		$	y_{e}^j = 
		\begin{cases}
			x_e & \text{ if $\ell(e) = j$ in $H$} \\
			1 & \text{ otherwise.}
		\end{cases}$
		%	\end{equation*}
		\item For $v_e \in V_E$, define $d_{e} = y_{e}^c$ where $c = \ell(e)$ in $H$.
		\item For $v \in V$, define $y_v^i = x_v^i$ for every $i \in [k]$.
		\item For $t_i \in T$ where $i \in [k]$, define $y_{t_i}^i = 0$, and $y_{t_i}^j = 1$ whenever $i \neq j$.
	\end{itemize}
	The objective for the \textsc{Node-MC} LP relaxation is then $\sum_{v_e \in V_e} w_e d_{e} = \sum_{e \in E} w_e x_e$, which matches the \minecc{} relaxation value. It only remains to check that the following {distance} constraints hold for all different types of edges $(u,v) \in \hat{E}$ and every color $j \in [k]$:
	\begin{align}
		\label{first}
		y_v^j &\leq y_u^j + d_v\\
		\label{secon}
		y_u^j &\leq y_v^j + d_u.
	\end{align}
	For $(v, v_e) \in E_H$, if $\ell(e) = c$ then $y_v^c = x_v^c \leq x_e = y_e^c = y_e^c + d_v$, since $d_v = 0$ for $v \in V$. Therefore, constraint~\eqref{first} holds in this case. Also, when $\ell(e) = c$, constraint~\eqref{secon} holds because $y_{e}^c = x_e = d_e \leq y_v^c + d_e$. 
	If $(v, v_e) \in E_H$ but we consider a color $i \neq c = \ell(e)$, then constraint~\eqref{first} is satisfied because $y_v^i = x_v^i \leq 1 = y_e^i = y_e^i + d_v$.
	Constraint~\eqref{secon} is satisfied as well since $y_e^i = 1 \leq x_v^i + x_v^c \leq x_v^i + x_e = y_v^i + d_e$. Here we have used the fact that $\sum_{i = 1}^k x_v^i = 1$ implies $1 \leq x_v^i + x_v^c$.
	Finally, for an edge $(v_e, t_c) \in E_T$ where $c = \ell(e)$, constraints~\eqref{first} and~\eqref{secon} become $y_e^j \leq y_{t_c}^j + d_e$ and  $y_{t_c}^j \leq y_e^j$.
	If $j \neq c$, then both of these hold because $y_{t_c}^j = y_e^j = 1$. Meanwhile, if $j = c$, then they follow from $y_e^j = d_e$ and $y_{t_i}^j = 0$.
\end{proof}

\paragraph{Strictly tighter instance.}
Theorem~\ref{thm:lps-app} shows that the \minecc{} relaxation is at least as tight of a lower bound as using the \textsc{Node-MC} relaxation. Although they have the same integrality gap, we can additionally show that the \minecc{} relaxation can be strictly tighter. Consider color set $C = \{1,2,3\}$ and let $H = (V,E)$ be a star graph with center node $v_0$ and three leaf nodes $\{v_1, v_2, v_3\}$, where edge $e_i = (v_0,v_i)$ has color $i$ for $i = 1,2,3$. The optimal \minecc{} solution will satisfy one of the edges and make mistakes at the two other edges. The \minecc{} LP relaxation will also have an optimal value of 2 by setting $x_{v_i}^i = 0$ for $i = 1,2,3$, and choosing either $x_{v_0}^i = \frac23$ for every $i = 1,2,3$, or $x_{v_0}^i = 1$ for exactly one $i \in \{1,2,3\}$. Meanwhile, if we reduce to an instance of \textsc{Node-MC}, each terminal node participates in one edge $(v_{e_i}, t_i)$ for $i \in \{1,2,3\}$. If we specifically consider the path constraints formulation of the \textsc{Node-MC} LP in~\eqref{eq:nwmc}, we can see that it is feasible to set $d_{e_i} = \frac12$, leading to an optimal LP solution of $\frac32$. This matches the worst case integrality gap.

\section{Additional Details for Vertex Cover Results}
\label{app:vc}
\begin{algorithm}
	\caption{$\textsf{PittVertexCover}(G)$}
	\label{alg:pitts}
	\begin{algorithmic}
		\STATE \textbf{Input:} $G = (V_G,E_G)$ with node weight $w_v \geq 0$ for each $v \in V_G$
		\STATE \textbf{Output:} Vertex cover $\mathcal{C} \subseteq V$.
		\STATE $\mathcal{C} \leftarrow \emptyset$  {\tt // Initialize empty vertex cover} 
		\FOR{$(u,v) \in E_G$}
		\IF{$u \notin \mathcal{C}$ and $v \notin \mathcal{C}$}
		\STATE Generate uniform random $\rho \in (0,1)$
		\IF{$\rho < \frac{w_v}{w_u + w_v}$}
		\STATE $\mathcal{C} \leftarrow \mathcal{C} \cup \{u\}$
		\ELSE 
		\STATE $\mathcal{C} \leftarrow \mathcal{C} \cup \{v\}$
		\ENDIF
		\ENDIF
		\ENDFOR
		\STATE Return $\mathcal{C}$
	\end{algorithmic}
\end{algorithm}
\begin{algorithm}
	\caption{$\textsf{PittColoring}(H)$}
	\label{alg:implicitpitts}
	\begin{algorithmic}
		\STATE{\bfseries Input:} Edge-colored hypergraph $H = (V,E,C, \ell)$ with weight $w(e) \geq 0$ for each $e \in E$
		\STATE {\bfseries Output:} Set $\mathcal{D} \subseteq E$ so that $H = (V, E-\mathcal{D},C,\ell)$ has no bad edge pairs
		\STATE $\mathcal{D} \leftarrow \emptyset$   
		\FOR{$v \in V$}
		\STATE $L_E(v) = [v(1) \;\; v(2) \;\; \cdots \;\; v(d_v)]$ {\tt // Indices for $v$'s edges, ordered by color}
		\STATE $f = 1$, $b = d_v$ \hspace{2.8cm}{\tt // Pointers to front/back of index array}
		%		\STATE {\tt // Cover all bad hyperedge pairs:}
		\STATE {\tt // Move past all covered bad edge pairs containing $v$}
		\WHILE{$e_{v(b)} \in \mathcal{D}$ and $b > f$}
		\STATE $b \leftarrow b-1$
		\ENDWHILE
		\WHILE{$e_{v(f)} \in \mathcal{D}$ and $b > f$}
		\STATE $f \leftarrow f+1$
		\ENDWHILE
		\STATE {\tt // Handle uncovered bad edge pairs containing $v$}
		\WHILE{$\ell({v(f)}) \neq \ell({v(b)})$}
		%		\STATE {\tt // $(e_{i(f)}, e_{i(b)})$ is a bad hyperedge pair}
		\STATE Generate $\rho \in (0,1)$
		\IF{$\rho < \frac{w({v(f)})}{w({v(f)}) + w({v(b)})}$}
		\STATE $\mathcal{D} \leftarrow \mathcal{D} \cup \{e_{v(b)}\}$
		\WHILE{$e_{v(b)} \in \mathcal{D}$ and $b > f$}
		\STATE {\tt // Ignore any deleted edges}
		\STATE $b \leftarrow b-1$
		\ENDWHILE
		\ELSE 
		\STATE $\mathcal{D} \leftarrow \mathcal{D} \cup \{e_{v(f)}\}$
		\WHILE{$e_{v(f)} \in \mathcal{D}$ and $b > f$}
		\STATE {\tt // Ignore any deleted edges}
		\STATE $f \leftarrow f+1$
		\ENDWHILE
		\ENDIF
		\ENDWHILE
		\ENDFOR
		\STATE Return $\mathcal{D}$
	\end{algorithmic}
\end{algorithm}

\subsection{Reduction proofs}
\textbf{Proof of Theorem~\ref{thm:vc2cec}}

\textit{Proof.}
	See Figure~\ref{fig:vcreductions}. To construct the edge-colored hypergraph $H$, introduce a node $v_{ij}$ for each $(i,j) \in E$, and for each $u \in V$ define a hyperedge $e_u = \{v_{uj} \colon j \text{ where $(u,j) \in E)$}\}$. Let each hyperedge be associated with its own unique color. 
	%	By construction, the degree of node $u \in V$ is the size of the hyperedge $e_u$. 
	Each node in $G$ corresponds to a hyperedge in $H$ and deleting a hyperedge in $H$ is equivalent to covering a node in $G$. Because all hyperedges have different colors, for every pair of overlapping hyperedges $(e_u, e_v)$ we know that either $e_u$ or $e_v$ must be deleted. This is equivalent to covering node $u$ or node $v$ when $(u,v) \in E$. By construction, the degree distribution in $G$ is exactly the hyperedge size distribution in $H$.
\hfill $\square$

\textbf{Proof of Theorem~\ref{thm:cec2vc}}

\textit{Proof.}
	See Figure~\ref{fig:vcreductions}.
	To construct the graph $G$, first introduce a node $v_e$ for every hyperedge $e \in E$. If two hyperedges $e \in E$ and $f \in E$ define a bad hyperedge pair 
	%	overlap and have different colors 
	($|e\cap f| > 0$ and $\ell(e) \neq \ell(f)$), then we add an edge $(v_e, v_f)$ in the graph $G$. If $e \in E$ has weight $w_e$, we give node $v_e$ weight $w_e$ in $G$. Deleting a minimum weight set of hyperedges in $H$ to ensure remaining hyperedges of different colors do not overlap is equivalent to deleting a minimum weight set of nodes in $G$ so that the remaining nodes form an independent set (i.e., the \textsc{Vertex Cover} problem).
\hfill $\square$

\subsection{Details and pseudocode for combinatorial algorithms}
Algorithm~\ref{alg:pitts} is pseudocode for Pitt's algorithm~\cite{pitt1985simple} for weighted \textsc{Vertex Cover}. Algorithm~\ref{alg:implicitpitts} is pseudocode for \textsf{PittColoring}, which is our 2-approximation for \minecc{} that implicitly applies Pitt's algorithm to the \textsc{Vertex Cover} instance $G$ that is approximation-equivalent to solving the \minecc{} objective on an edge colored hypergraph $H = (V,E,C, \ell)$. Importantly, this algorithm never explicitly forms $G$, but is still able to identify a set of edges in $H$ to delete which correspond to a vertex cover in $G$. Full details for how to accomplish this is included in the main text, along with justification for the $O(\sum_{e \in E} |e|)$ runtime. The fact that this is a 2-approximation for \minecc{} (even in the weighted case) follows directly from the fact that this is implicitly providing a valid $2$-approximate vertex cover in the (unformed) reduced graph $G$ via Pitt's algorithm.

\paragraph{\textsf{MatchColoring}.}
Our algorithm \textsf{MatchColoring} is a 2-approximation for the \emph{unweighted} version of \minecc{}. This algorithm implicitly applies the common strategy of approximating an unweighted \textsc{Vertex Cover} problem by finding a maximal matching and then adding both endpoints of each edge in the maximal matching to form a cover. For our problem, this corresponds to finding a maximal edge-disjoint set of bad edge pairs, and deleting all edges in those bad edge pairs. This can also be implemented in $O(\sum_e |e|)$ time by carefully iterating through bad edge pairs in the same way that \textsf{PittColoring} does. This approach has the additional advantage of providing a deterministic 2-approximate solution for unweighted \minecc{} (\textsf{PittColoring} has an \emph{expected} 2-approximation). Additionally, the maximal matching (i.e., maximal edge-disjoint set of bad edge pairs) that it computes provides a lower bound on the \minecc{} objective, which can be used in practice to check a posteriori approximation guarantees for other methods that do not come with approximation guarantees of their own.

\paragraph{Improved \textsf{Hybrid} algorithm.}
Our \textsf{Hybrid} algorithm combines the strengths of \textsf{MatchColoring} and \textsf{MajorityVote}. \textsf{MatchColoring} finds a maximal edge-disjoint set of bad hyperedge pairs, which provides a useful lower bound on the optimal \minecc{} objective. It deletes all edges in this maximal set, which is exactly within a factor 2 of this lower bound. In practice, this often deletes more edges than is strictly necessary, and leads to a large set of nodes that are contained in no remaining edge. From a theoretical perspective, these can be assigned any color and the algorithm will still be a 2-approximation. \textsf{Hybrid} works simply by assigning these nodes to have the color determined by the \textsf{MajorityVote} method (i.e, the edge color that the node participates in the most). This often ends up satisfying many edges that were deleted by \textsf{MatchColoring} even though they did not actually need to be deleted in order to yield a hypergraph with no bad edge pairs. Because the first step of \textsf{Hybrid} is to run \textsf{MatchColoring} and assign colors based on this algorithm, in theory is still enjoys the theoretical 2-approximation guarantee, and it can still make use of the explicit lower bound computed by this \textsf{MatchColoring}.

\section{Experimental Results}
\label{sec:appexp}
%\begin{table*}[t!]
%	\caption{Summary.}
%	\label{tab:allexp1}
%	\centering
%	\scalebox{0.95}{\begin{tabular}{l     l l l l     l l l l l     l l l l l    l l l ll}
%			\toprule
%			&&&&& \multicolumn{5}{c}{\textbf{Approx. Guarantee}} & \multicolumn{5}{c}{\textbf{Edge Satisfaction}} & \multicolumn{5}{c}{\textbf{Runtime (in seconds)}} \\ 
%			\cmidrule(lr){6-10} \cmidrule(lr){11-15}  \cmidrule(lr){16-20}
%			%\!\!\!\!\!\!\!
%			\emph{Dataset} & $\lvert V \rvert$ & $\lvert E \rvert$ & $r$ & $k$ & LP  & MV   & IC   & CB   & LCB  & LP   & MV   & IC   & CB   & LCB  & LP & MV& IC & CB & LCB \\
%			\midrule
%			Brain & 638 & 21180 & 2 & 2 &1.0 & 1.01 & 1.27 & 1.56 & 1.41 & 0.64 & 0.64 & 0.55 & 0.44 & 0.5 & 1.8 & 0.0 & 1.9 & 0.4 & 0.8\\
%			MAG-10 & 80198 & 51889 & 25 & 10 &1.0 & 1.18 & 1.37 & 1.44 & 1.35 & 0.62 & 0.55 & 0.48 & 0.45 & 0.49 & 51 & 0.1 & 203 & 333 & 699\\
%			Cooking & 6714 & 39774 & 65 & 20 &1.0 & 1.21 & 1.21 & 1.23 & 1.24 & 0.2 & 0.03 & 0.03 & 0.01 & 0.01 & 72 & 0.0 & 1223 & 4.6 & 6.7\\
%			DAWN & 2109 & 87104 & 22 & 10 &1.0 & 1.09 & 1.0 & 1.31 & 1.15 & 0.53 & 0.48 & 0.53 & 0.38 & 0.46 & 13 & 0.0 & 190 & 0.3 & 0.4\\
%			Walmart-Trips & 88837 & 65898 & 25 & 44 & 1.0 & 1.2 & 1.19 & 1.26 & 1.26 & 0.24 & 0.09 & 0.09 & 0.04 & 0.05 & 7686 & 0.2& 68801 & 493 & 1503\\
%			\bottomrule
%	\end{tabular}}
%\end{table*} 

\begin{table}[t]
	\caption{Statistics for five benchmark edge-colored hypergraphs from~\citet{amburg2020clustering}.}
	\label{tab:size-run}
	\centering
%	\scalebox{0.95}{
		\begin{tabular}{l     l l l l }
			\toprule
			\emph{Dataset} & $\lvert V \rvert$ & $\lvert E \rvert$ & $r$ & $k$ \\
			\midrule
\texttt{Brain} & 638 & 21180 & 2 & 2 \\
\texttt{Cooking} & 6714 & 39774 & 65 & 20  \\
\texttt{DAWN} & 2109 & 87104 & 22 & 10  \\
\texttt{MAG-10} & 80198 & 51889 & 25 & 10  \\
\texttt{Walmart-Trips} & 88837 & 65898 & 25 & 44 \\
			\bottomrule
	\end{tabular}
\end{table} 

\begin{table*}[t]
	\caption{Approximation factors (ratio between algorithm output and LP lower bound), edge satisfaction (percentage of edges satisfied by the clustering), and runtimes obtained by various algorithms on five benchmark edge-colored hypergraphs from~\citet{amburg2020clustering}.
		\textsf{MajorityVote} (MV) is deterministic. Our vertex cover algorithms are randomized, so we list the mean values and standard deviations over 50 runs. \textsf{Pitt+} and \textsf{Match+} correspond to running \textsf{PittColoring} and \textsf{MatchColoring} 100 times and taking the best clustering found. 
	}
	\label{tab:approx}
	\centering
%	\scalebox{0.95}{
		\begin{tabular}{l     l l l l l l }
			\toprule
			& \multicolumn{6}{c}{\textbf{Ratio to LP lower bound}}  \\ 
			\cmidrule(lr){2-7} \
			%\!\!\!\!\!\!\!
			\emph{Dataset} & \textsf{LP} & \textsf{MV} & \textsf{PittColoring} & \textsf{MatchColoring} &\textsf{Pitt+} & \textsf{Match+} \\
			\midrule
			\texttt{Brain}  & 1.0 & 1.01 & 1.07 {\small $\pm 0.01$} & 1.08 {\small $\pm 0.01$}& 1.06 {\small $\pm 0.01$}& 1.07 {\small $\pm 0.01$} \\
			\texttt{Cooking}  & 1.0 & 1.21 & 1.23 {\small $\pm 0.01$} & 1.23 {\small $\pm 0.0$}& 1.22 {\small $\pm 0.0$}& 1.22 {\small $\pm 0.01$} \\
			\texttt{DAWN}  & 1.0 & 1.09 & 1.57 {\small $\pm 0.04$} & 1.58 {\small $\pm 0.03$}& 1.54 {\small $\pm 0.03$}& 1.54 {\small $\pm 0.03$} \\
			\texttt{MAG-10}  & 1.0 & 1.18 & 1.39 {\small $\pm 0.01$} & 1.49 {\small $\pm 0.0$}& 1.37 {\small $\pm 0.0$}& 1.48 {\small $\pm 0.0$} \\
			\texttt{Walmart-Trips}  & 1.0 & 1.2 & 1.13 {\small $\pm 0.0$} & 1.18 {\small $\pm 0.0$}& 1.13 {\small $\pm 0.0$}& 1.17 {\small $\pm 0.0$} \\
			\toprule
			& \multicolumn{6}{c}{\textbf{Edge Satisfaction}}  \\ 
			\cmidrule(lr){2-7} \
			%\!\!\!\!\!\!\!
			\emph{Dataset} & \textsf{LP} & \textsf{MV} & \textsf{PittColoring} & \textsf{MatchColoring} &\textsf{Pitt+} & \textsf{Match+} \\
			\midrule
		\texttt{Brain}   & 0.64 & 0.64 & 0.62 {\small $\pm 0.0$} & 0.62 {\small $\pm 0.0$}& 0.62 {\small $\pm 0.0$}& 0.62 {\small $\pm 0.0$}\\
		\texttt{Cooking}   & 0.2 & 0.03 & 0.01 {\small $\pm 0.0$} & 0.01 {\small $\pm 0.0$}& 0.02 {\small $\pm 0.0$}& 0.02 {\small $\pm 0.0$}\\
		\texttt{DAWN}   & 0.53 & 0.48 & 0.26 {\small $\pm 0.02$} & 0.25 {\small $\pm 0.02$}& 0.27 {\small $\pm 0.01$}& 0.27 {\small $\pm 0.01$}\\
		\texttt{MAG-10}   & 0.62 & 0.55 & 0.47 {\small $\pm 0.0$} & 0.44 {\small $\pm 0.0$}& 0.48 {\small $\pm 0.0$}& 0.44 {\small $\pm 0.0$}\\
		\texttt{Walmart-Trips}   & 0.24 & 0.09 & 0.14 {\small $\pm 0.0$} & 0.11 {\small $\pm 0.0$}& 0.14 {\small $\pm 0.0$}& 0.11 {\small $\pm 0.0$}\\
		\toprule
			& \multicolumn{6}{c}{\textbf{Runtime}}  \\ 
			\cmidrule(lr){2-7} \
			\emph{Dataset} & \textsf{LP} & \textsf{MV} & \textsf{PittColoring} & \textsf{MatchColoring} &\textsf{Pitt+} & \textsf{Match+} \\
			\midrule
			\texttt{Brain}  &  0.52 & 0.001 & 0.006 {\small $\pm 0.028$} & 0.002 {\small $\pm 0.001$}& 0.12 {\small $\pm 0.006$}& 0.028 {\small $\pm 0.004$} \\
			\texttt{Cooking}  &  127.01 & 0.002 & 0.01 {\small $\pm 0.003$} & 0.008 {\small $\pm 0.007$}& 0.525 {\small $\pm 0.012$}& 0.228 {\small $\pm 0.011$} \\
			\texttt{DAWN}  &  4.23 & 0.003 & 0.01 {\small $\pm 0.004$} & 0.005 {\small $\pm 0.003$}& 0.779 {\small $\pm 0.038$}& 0.221 {\small $\pm 0.007$} \\
			\texttt{MAG-10}  &  17.0 & 0.012 & 0.04 {\small $\pm 0.006$} & 0.04 {\small $\pm 0.016$}& 0.675 {\small $\pm 0.037$}& 0.414 {\small $\pm 0.06$} \\
			\texttt{Walmart-Trips}  &  321.09 & 0.07 & 0.053 {\small $\pm 0.007$} & 0.05 {\small $\pm 0.009$}& 1.545 {\small $\pm 0.083$}& 1.044 {\small $\pm 0.067$} \\
			\bottomrule
	\end{tabular}
%}
\end{table*} 

Previous work has already shown that the LP relaxation can often be solved on real-world graphs and hypergraphs with tens of thousands of nodes and many large hyperedges within minutes~\cite{amburg2020clustering,amburg2022diverse}. The LP relaxation often even finds the optimal solution even using the simplest rounding techniques. Our improved LP approximation results therefore mainly serve as a way bring the best theoretical results closer to what has been observed in practice. On the other hand, our combinatorial algorithms provide a very practical new way to obtain approximate solutions with strong guarantees on a much larger scale than was ever possible previously. We implement our combinatorial algorithms in Julia. All experiments are run on a Mac laptop with 16GB of RAM. Source code and all data is publicly available on the Github repo \url{https://github.com/nveldt/ImprovedECC}.
%An anonymized version of our code and the \emph{Trivago} dataset are provided at~\url{tinyurl.com/bdd7ns54}. A de-anonymized version will be made publicly available if the manuscript is accepted. 

\paragraph{Experiments on previous benchmark hypergraphs}
Table~\ref{tab:size-run} presents statistics for five benchmark edge-colored clustering hypergraphs first considered by~\citet{amburg2020clustering}. In Table~\ref{tab:approx}, we show the performance of various algorithms in approximating ECC on these instances. The LP relaxation produces integral or near integral results on these datasets, so it is enough to apply a very simple rounding scheme (for each $v \in V$, assign it color $i^* = \argmin_i x_v^i$) to obtain an optimal or near optimal solution to \minecc{}. \textsf{PittColoring} and \textsf{MatchColoring} exhibit a very straightforward tradeoff when compared with LP: they are far more scalable, at the expense of worse approximation guarantees. They nevertheless produce clusterings that are within a small factor of the LP lower bound (Approximation Factor in Table~\ref{tab:approx}), and only take a fraction of a second.

The comparison between our combinatorial methods and \textsf{MajorityVote} is arguably more interesting. Although \textsf{MajorityVote} only has an $r$-approximation guarantee in theory, it often produces better clusterings than our combinatorial methods. However, on the \texttt{Walmart} dataset, \textsf{PittColoring} and \textsf{MatchColoring} produce better results. This may be due to the fact that this dataset has a much larger maximum hyperedge size, though it is not entirely clear. All of these combinatorial methods have the same asymptotic runtime of $O(\sum_e |e|)$, and very similar running times in practice. We have also reported results for running \textsf{PittColoring} and \textsf{MatchColoring} 100 times and taking the best result (\textsf{Pitt+} and  \textsf{Match+} in Table~\ref{tab:approx}). This is still very fast, and produces slightly better approximation factors. For each run, we randomize the order in which nodes are visited, which changes the order in which bad edge pairs are visited and therefore changes which hyperedges are deleted. \textsf{PittColoring} has additional randomization in how it chooses which edge to delete in a bad edge pair.

\paragraph{Experiments on the Trivago Hypergraph}
The \texttt{Trivago} hypergraph is derived from the 2019 ACM RecSys Challenge dataset (\url{https://recsys.acm.org/recsys19/challenge/}). It is closely related to the previous Trivago hypergraph without edge labels available at \url{https://www.cs.cornell.edu/~arb/data/trivago-clicks/}. The key difference is that in processing the data, we have kept track of the user location platform when parsing the website browsing data, which allows us to obtain edge country labels. We discard nodes that have labels that do not correspond to any of the hyperedge labels. 
%Complete details and code for downloading and converting the original dataset into an edge-colored hypergraph will be released with the final version of the code. 

Table~\ref{tab:trivago} in the main text displays the number of mistakes, edge satisfaction, a posteriori approximation ratio (for all but \textsf{PittColoring}, which does not compute an explicit lower bound), accuracy, and runtime of each combinatorial method, all averaged over 50 different trials. There is little to no variation between different runs. The a posteriori approximation guarantee for \textsf{MatchColoring} and \textsf{Hybrid} is obtained by taking the number of mistakes made by the method and dividing by the lower bound computed by \textsf{MatchColoring}. \textsf{MajorityVote} does not compute as good of a lower bound. In theory, this provides an $r$-approximation, because this method can be seen as the optimal solution to an alternative edge-colored clustering objective where the penalty at each hyperedge is the \emph{number} of nodes in the hyperedge with a color that is different from the hyperedge's color. This is always within a factor $r$ of the optimal \minecc{} objective. Using this fact, it is not hard to prove that the optimal \minecc{} objective will be lower bounded by
\begin{equation}
	\label{lower}
\frac{\sum_{e \in E} \sum_{v \in e} \mathbbm{1} \left( Y_\text{MV}[v] \neq \ell(e) \right)}{r},
\end{equation}
where $Y_{MV}$ is the \textsf{MajorityVote} clustering. Dividing the \minecc{} objective for \textsf{MajorityVote} by this lower bound on the optimal \minecc{} solution produces the approximation in Table~\ref{tab:trivago}. Although this is better than the theoretical 85-approximation, it is still very poor. 

%For this \textit{Trivago} hypergraph, solving and rounding the LP relaxation does produce a very good output, with an a posteriori approximation guarantee of $1.0014$, an edge satisfaction of $0.765$, and an accuracy of $0.80$. However, it takes roughly 25 minutes to solve the LP relaxation, which is roughly four orders of magnitude slower than our linear-time combinatorial algorithms. This result does indicate that LP-based methods for \minecc{} are very strong and very successful when it is possible to run them. However, our combinatorial algorithms will be able to scale to extremely large datasets where it is infeasible to rely on LP-based techniques.

%%%%%%%%%%%%%%%%%%%%%%%%%%%%%%%%%%%%%%%%%%%%%%%%%%%%%%%%%%%%%%%%%%%%%%%%%%%%%%%
%%%%%%%%%%%%%%%%%%%%%%%%%%%%%%%%%%%%%%%%%%%%%%%%%%%%%%%%%%%%%%%%%%%%%%%%%%%%%%%

\end{document}